\documentclass[english,hyphens]{article}
\usepackage[T1]{fontenc}
\usepackage[utf8]{inputenc}
\usepackage{babel}
\usepackage{refstyle}
\usepackage{float}
\usepackage{calc}
\usepackage{textcomp}
\usepackage{amsthm}
\usepackage{amsmath}
\usepackage{amssymb}
\usepackage{graphicx}
\usepackage[unicode=true]
 {hyperref}

\makeatletter

\AtBeginDocument{\providecommand\secref[1]{\ref{sec:#1}}}
\providecommand{\tabularnewline}{\\}
\floatstyle{ruled}
\newfloat{algorithm}{tbp}{loa}
\providecommand{\algorithmname}{Algorithm}
\floatname{algorithm}{\protect\algorithmname}
\RS@ifundefined{subref}
  {\def\RSsubtxt{section~}\newref{sub}{name = \RSsubtxt}}
  {}
\RS@ifundefined{thmref}
  {\def\RSthmtxt{theorem~}\newref{thm}{name = \RSthmtxt}}
  {}
\RS@ifundefined{lemref}
  {\def\RSlemtxt{lemma~}\newref{lem}{name = \RSlemtxt}}
  {}

\numberwithin{equation}{section}
\numberwithin{figure}{section}
\theoremstyle{plain}
\newtheorem{thm}{\protect\theoremname}

\usepackage{hyperref}
\usepackage[T1]{fontenc}
\usepackage{microtype}

\makeatother

\providecommand{\theoremname}{Theorem}

\begin{document}

\title{\textbf{Raziel}:\\
Private and Verifiable Smart Contracts on Blockchains}

\author{David Cerezo Sánchez%
\thanks{Corresponding author%
}\\
{\small{}Calctopia\textsuperscript{{*}}}\\
{\small{}david@calctopia.com\textsuperscript{{*}}}\\
{\small{}}\\
{\small{}}\\
{\small{}(This paper is a first implementation of \href{https://www.calctopia.com/papers/csfi.pdf}{PCT/IB2015/055776})}\\
{\small{}}\\
}
\maketitle
\begin{abstract}
Raziel combines secure multi-party computation and proof-carrying
code to provide privacy, correctness and verifiability guarantees
for smart contracts on blockchains. Effectively solving DAO and Gyges
attacks, this paper describes an implementation and presents examples
to demonstrate its practical viability (e.g., private and verifiable
crowdfundings and investment funds, double auctions for decentralized
exchanges). Additionally, we show how to use Zero-Knowledge Proofs
of Proofs (i.e., Proof-Carrying Code certificates) to prove the validity
of smart contracts to third parties before their execution without
revealing anything else. Finally, we show how miners could get rewarded
for generating pre-processing data for secure multi-party computation.
\pagebreak{}
\end{abstract}

\section{Introduction}

The growing demand for blockchain and smart contract\cite{ideaSmartContracts,1703.06322,1608.00771,Jones_howto,Hvitved10asurvey,cryptoeprint:2016:1156}
technologies sets the challenge of protecting them from intellectual
property theft and other attacks\cite{cryptoeprint:2016:1007}: security,
confidentiality and privacy are the key issues holding back their
adoption\cite{americannBankerPrivacy}. As shown in this paper, the
solution must be inter-disciplinary (i.e., cryptography and formal
verification techniques): code from smart contracts cannot be formally-verified
before its execution with cryptographic techniques and conversely,
formal verification techniques keep track of information flows to
guarantee confidentiality but they don't provide privacy on their
own\cite{Fournet:2009:SCD:1653662.1653715,Fournet:2011:ITH:2046707.2046747}.

The availability of the proposed technical solution and its broad
adoption is ultimately a common good: it's in the public interest
to use smart contracts that respect the confidentiality and privacy
of the processed data and that can be efficiently verified. Smart
contracts have recently been heralded as ``cryptocurrency's killer
app''\cite{smartContractKillerApp}, but to develop the next digital
businesses based on the blockchain (e.g., the equivalents of Paypal,
Visa, Western Union, NYSE) there is an urgent need for better protected
smart contracts, a gap solved on this publication.

\paragraph{Contributions}

We propose a system for securely computing smart contracts guaranteeing
their privacy, correctness and verifiability. Our main and novel contributions
are:
\begin{itemize}
\item Practical formal verification of smart contracts: the proofs accompanying
the smart contracts can be used to prove functional correctness of
a computation as well as other properties including termination, security,
pre-conditions and post-conditions, invariants and any other requirements
for well-behaved code. A smart contract that has been fully formally
verified protects against the subtle bugs such as the one that enabled
millions to be stolen from the DAO\cite{daoanalysis}. And when signed
proofs conform to the specifications of a trusted organization, Raziel
also prevents the execution of criminal smart contracts (i.e., Gyges
attacks\cite{EPRINT:JueKosShi16}) solving this kind of attacks for
the first time.
\item A practical use of Non-Interactive Zero-Knowledge proofs of proofs/certifiable
certificates is proposed (see \secref{ZK-Proofs}): the code producer
can convince the executing party of the smart contract of the existence
and validity of proofs about the code without revealing any actual
information about the proofs themselves or the code of the smart contract.
Although the concept of zero-knowledge proofs of proofs is not new\cite{Blum87howto,BGG90,Pope04provinga},
it can now be realised in practice for general properties of code.
\item An outsourcing protocol for secure computation that allows offline
parties and private parameter reuse is presented (see \secref{outsourcing}):
previous approaches don't allow all the parties to be offline\cite{reuseloseit,182943,EPRINT:CarLevTra14,EPRINT:CarTra16,cryptoeprint:2014:596}
or reusing encrypted values\cite{EPRINT:KamMohRay11,EPRINT:KamMohRiv12}.
\item We propose that miners should get rewarded for generating pre-processing
data for secure multi-party computation (see \secref{miningPreprocess}),
in line with earlier attempts to replace wasteful PoWs\cite{primecoin,permacoin,cryptoeprint:2017:203,pownp}.
\end{itemize}

\paragraph{Minimal set of functionalities}

We argue that the combined features here described are indeed the
minimal set of functionalities that must be offered by a secure solution
to protect computations on blockchains:
\begin{itemize}
\item A pervasive goal of blockchain technologies is the removal of trusted
third parties: towards this end, the preferred solution in cryptography
is secure multi-party computation, which covers much more than transactions.
Moreover, public permissionless blockchains are executing smart contracts
out in the open and any curious party is able to inspect the input
parameters and detailed execution of any on-chain smart contract:
again, MPC is the best cryptographic solution to protect the privacy
of the computation of said smart contracts.

\begin{itemize}
\item Since carrying out encrypted computations has a high cost, we introduce
two improvements to reduce that cost: outsourcing for cloud-based
blockchains (see \secref{outsourcing}) and mining pre-processing
data for secure multi-party computation (see \secref{miningPreprocess}).
\end{itemize}
\item Once encrypted smart contracts are being considered, a moral hazard
arises: why should anyone execute a potentially-malicious encrypted
smart contract from an anonymous party? Formal verification techniques
such as proof-carrying code provide an answer to this dilemma: mathematical
proofs about many desirable properties (termination, correctness,
security, resource consumption, legal/regulatory, economic and functional,
among others) can be offered before carrying out any encrypted execution.

\begin{itemize}
\item To prevent that said proofs leak too much information about the smart
contracts, a novel solution for zero-knowledge proofs of proofs is
proposed (see \secref{ZK-Proofs}).
\end{itemize}
\end{itemize}
This paper describes an implementation discussing cryptographic and
technical trade-offs.

\section{Background}

This section provides a brief introduction to the main technologies
that underpin Raziel: blockchains, secure multi-party computation\cite{shafiMPC}
and formal verification.

\paragraph{Blockchains}

A blockchain is a distributed ledger that stores a growing list of
unmodifiable records called blocks that are linked to previous blocks.
Blockchains can be used to make online secure transactions, authenticated
by the collaboration of the P2P nodes allowing participants to verify
and audit transactions. Blockchains can be classified according to
their openness. Open, permissionless networks don't have access controls
(e.g., Bitcoin\cite{bitcoin} and Ethereum\cite{ethereum,yellowEthereum})
and reach decentralized consensus through costly Proof-of-Work calculations
over the most recently appended data by miners. Permissioned blockchains
have identity systems to limit participation (e.g., Hyperledger Fabric\cite{cachin2016architecture})
and do not depend on proofs-of-work. Blockchain-based smart contracts
are computer programs executed by the nodes and implementing self-enforced
contracts. They are usually executed by all or many nodes (\textit{on-chain
smart contracts}), thus their code must be designed to minimize execution
costs. Lately, off-chain smart contracts frameworks are being developed
that allow the execution of more complex computational processes.

\paragraph{Secure Multi-Party Computation}

Protocols for secure multi-party computation (MPC) enable multiple
parties to jointly compute a function over inputs without disclosing
said inputs (i.e., secure distributed computation). MPC protocols
usually aim to at least satisfy the conditions of inputs privacy (i.e.,
the only information that can be inferred about private inputs is
whatever can be inferred from the output of the function alone) and
correctness (adversarial parties should not be able to force honest
parties to output an incorrect result). Multiple security models are
available: semi-honest, where corrupted parties are passive adversaries
that do not deviate from the protocol; covert, where adversaries may
deviate arbitrarily from the protocol specification in an attempt
to cheat, but do not wish to be ``caught'' doing so ; and malicious
security, where corrupted parties may arbitrarily deviate from the
protocol. Multiple related cryptographic techniques (including secret
sharing\cite{shamirSecretSharing}, oblivious transfer\cite{EPRINT:Rabin05},
garbled circuits, oblivious random access machines\cite{Goldreich:1987:TTS:28395.28416})
and MPC protocols (e.g., Yao\cite{yao82,yao86,cryptoeprint:2004:175},
GMW\cite{gmw1987}, BMR\cite{bmr}, BGW\cite{BGW88} and others) have
been developed since MPC was originally envisioned by Andrew Yao in
the early 1980s.

\paragraph{Formal verification}

Formal verification uses formal methods of mathematics on software
to prove or disprove its correctness with respect to certain formal
specifications. Deductive verification is the preferred approach:
smart contracts are annotated to generate proof obligations that are
proved using theorem provers or satisfiability modulo theories solvers.
In the software industry, formal verification is hardly used, but
in hardware the high costs of recalling defective products explain
their greater use: analogously, it's expected a higher use of formal
verification methods applied to smart contracts compared to the general
rate of use on the software industry given the losses in case of errors
and exploitation\cite{daoanalysis}.

\section{Model and Goals}

Parties executing smart contracts need to protect their private financial
information and obtain formal guarantees regarding their execution.
The central goal for Raziel is to offer a programming framework to
facilitate the development of formally-verified, privacy-preserving
smart contracts with secure computation.

\subsection{Threat Model and Assumptions}

A conservative threat model assumes that parties wish to execute smart
contracts but mutually distrust each another. Each party is potentially
malicious and the smart contract is developed by one of the parties
or externally developed. We assume that each party trusts its own
environment and the blockchain; the rest of the system is untrusted.
The threat model does not include side channel attacks or denial of
service attacks.

\subsection{Goals}

A secure smart contract system should operate as follows: a restricted
set of parties willing to execute a smart contract check the accompanying
proofs/certificates to verify it before its execution. Then these
parties send their private inputs to the nodes at the start of the
execution of the smart contract; after that, the nodes run the smart
contract carrying out the secure computation. Alternatively, the same
parties may run the secure computation between themselves without
the use of external executing nodes. Finally, the restricted set of
parties obtain the results from the secure computation, which could
be stored on the blockchain according to the rules of some consensus
protocol.

\section{Private Smart Contracts\label{sec:Private-Smart-Contracts}}

Enabling smart contracts with secure computation techniques (e.g.,
secure multi-party computation, homomorphic encryption\cite{Rivest1978},
indistinguishability obfuscation\cite{EPRINT:BitVai15}) is a key-step
to the global adoption of blockchain technologies: encrypted transactions
could be stored on the blockchain; secure computations could be carried
out between distrustful parties; even the contract's code on the blockchain
could be kept private. 

Although it would be very profitable to sell/acquire smart contracts
based on the value of their secret algorithms (see \nameref{sub:Markets-for-Smart-Contracts})
using homomorphic encryption (bootstrapping being a costly operation
that can only be minimized\cite{minBootstrappings,EPRINT:BLMZ16}
and not avoided) and indistinguishability obfuscation (polynomial-time
computable, but with constant factors $\geq2^{100}$), are both considered
currently infeasible due to concrete efficiency issues. In spite of
impressive progress towards making these techniques practical\cite{EPRINT:LMACFW16,cryptoeprint:2017:826,cryptoeprint:2017:430},
they appear to be a long way from suitability for our purposes (i.e.,
\textit{\small{}Obfustopia} is still a very expensive place, \textit{\small{}Cryptomania}
is much more affordable).

Another possible approach, which is being adopted by some blockchains
and related technologies, is to rely on a trusted execution environment
(most notably Intel's SGX\cite{cryptoeprint:2016:086}). Relying on
trusted hardware is a risky bet and assumes a high level of trust
in the hardware vendor. Several severe attacks have exposed vulnerabilities
of SGX\cite{1702.07521,1702.08719,1703.06986,asyncshock,Xu:2015:CAD:2867539.2867677,tsgx,1611.06952,rollbackSGX,cryptoeprint:2017:736,203696,intelSGXBug,SGXnoPageFaults,jang:sgx-bomb,DBLP:journals/corr/XiaoLCZ17,leakyCauldron,SGXStep,spectreSGX,1802.09085,cachequote,vanbulck2018foreshadow,guarddilemma,nemesis,2018arXiv181105441C,splitSpectre,SMoTherSpectre,ZombieLoad2019,ridl}:
all the current proposed and existing blockchains whose security rest
on SGX aren't providing detailed explanations and proofs on how they
are defending against these attacks. Unlike software bugs, new hardware
would have to be deployed to fix these kinds of bugs, sometimes re-architecting
the full solution.

Hence, our design employs secure multi-party computation. Secure multi-party
computation is more mature than the fully homomorphic methods, and
has a less trusting threat model than trusted execution approaches.
MPC protocols have solid security proofs based on standard assumptions
and efficient implementation. Current MPC technologies that can be
used in production\cite{EPRINT:ABPP15} include: secret sharing, garbled
circuits, oblivious transfer and ORAMs. Secret sharing based schemes
(i.e., many-round, dependent on the depth of the circuit) can be faster
than garbled circuits/BMR (i.e., constant-round) in low-latency settings
(LANs): when latency is large or unknown, it's better to use constant-round
protocols\cite{gmwvsyao,cryptoeprint:2016:1066}.

A preferred approach is to use modular secure computation frameworks:
for security reasons, at least two secure computation frameworks should
be available, offering different cryptographic and security assumptions;
in case the security of one of them is compromised, there will be
a second option.

\subsection{Off-chain computation}

Even when using the fastest available secure multi-party computation
techniques, the overhead would be very significant if secure computations
would have to be executed on every full node of a blockchain, as have
been previously proposed\cite{buterinSecretDAO}. In Ethereum, at
current prices (380 \$/ETH, 21 Gwei/gas, 30/August/2017), multiplying
or dividing 2 plaintext integers 10 million times costs:
\[
5\frac{\mbox{Gas}}{\mbox{ops}}\cdot10000000\,\mbox{ops}\cdot0.000000001\frac{\mbox{ETH}}{\mbox{Gwei}}\cdotp21\frac{\mbox{Gwei}}{\mbox{Gas}}\cdot380\frac{\mbox{\$}}{\mbox{ETH}}=\$399
\]
and to store 32768 words of 256 bits (i.e. 1 megabit), it costs:
\[
2000\frac{\mbox{Gas}}{\mbox{SSTORE}}\cdot32768\,\mbox{ops}\cdot0.000000001\frac{\mbox{ETH}}{\mbox{Gwei}}\cdotp21\frac{\mbox{Gwei}}{\mbox{gas}}\cdot380\frac{\mbox{\$}}{\mbox{ETH}}=\$523
\]

However, multiplying 2 plaintext integers 10 million times on a modern
computer takes 0.02 seconds: since it costs \$0.004/hour (Amazon EC2
t2.nano 1-year Reserved Instance), the same multiplications cost
\[
\frac{\$0.004\,\mbox{hour}}{3600\,\mbox{seconds/hour}}\cdot0.02\,\mbox{seconds}=\$0.000000022
\]
It's order of magnitude more expensive, concretely
\[
\frac{\$399}{\$0.000000022}=18136363636.363636364
\]

That is, more than 18 billions times more expensive: and as the price
of Ether keeps rising, the costs of computation will also rise. On
another note, current state-of-the-art secure computation executes
7 billion AND gates per second on a LAN setting \cite{EPRINT:AFLNO16},
but an 8-core processor executes 316 GIPS at the Dhrystone Integer
Native benchmark (Intel Core i7-5960X), an approximate slowdown of
\[
\frac{316\,\mbox{GIPS}}{7\,\mbox{Billion gates/second}}=45.142857143\,\mbox{x}
\]
which must not be additionally imposed to the overhead/overcost of
a public permissionless distributed ledger: on the hand, it also means
that there are overheads/overcosts being accepted which are higher
than those imposed by secure computation or verifiable computation
($10^{5}-10^{7}$ for proving correctness of the execution \cite{Walfish:2015:VCW:2728770.2641562}).

We address on-chain and off-chain secure computations separately:
\begin{itemize}
\item For on-chain secure computations, the validation of transactions happens
through the replicated execution of the smart contract and given the
fault assumption underlying the consensus algorithm. Because privacy-preserving
smart contracts introduce a significant resource consumption overhead,
it is important to prevent the execution of the same privacy-preserving
smart contract on every node. We relax the full-replication requirement
for PBFT consensus: a “query” command could be executed on a restricted
set of nodes (large enough to provide consensus), and then an “invoke”
store the result in the distributed ledger. Finally, achieving consensus
is a completely deterministic procedure, with no possible way for
differences.
\item Regarding off-chain secure computations, simple oracle-like calls
could be considered or other more complex protocols\cite{aeternity,truebit,iEx.ec}
for scalable off-chain computation.
\end{itemize}

\subsubsection{Latency and its impact}

As previously mentioned, MPC protocols can be roughly divided into
two classes: constant-round protocols, ideal for high latency settings;
and protocols with rounds dependent on the depth of the evaluated
circuit, usually faster but only on very low latency settings. The
present paper proposes the use of two protocol suites to get the maximum
performance, independent of the latency:
\begin{itemize}
\item if all parties are on a low latency setting, they could use a very
fast secret-shared protocol\cite{EPRINT:AFLNO16,EPRINT:BogLauWil08}
\item but if parties are on a high latency setting, they must use a constant-round
protocol\cite{cryptoeprint:2017:189,emp-toolkit,cryptoeprint:2016:1066,cryptoeprint:2017:214,cryptoeprint:2017:862,improvementsBMR,cryptoeprint:2018:208}.
Alternatively, they could outsource the secure computation to a cloud
setting (see \secref{Appendix-outsource}), to use a faster secret-shared
protocol.
\end{itemize}

\subsection{Blockchain Solutions}

The guarantees of privacy, correctness and verifiability are designed
for the more threatening setting of public permissionless blockchains,
although the smart contracts can also be used on private permissioned
blockchains: it's also preferred that private blockchains keep their
communications open to public blockchains to allow the use of smart
contracts of public utility.

We considered three blockchains solutions:
\begin{itemize}
\item Cothority\cite{cothority}/ByzCoin\cite{byzcoin}/OmniLedger\cite{cryptoeprint:2017:406}
is a next-generation permissionless distributed ledger designed for
high-throughtput while preserving long-term security, supporting Visa-level
workloads and confirming typical transactions in under two seconds. 
\item Hyperledger Fabric\cite{cachin2016architecture} is a good fit for
executing complex computational procedures like secure multi-party
computation and formal verification techniques: including such complex
procedures on its chaincodes (i.e. smart contracts) requires no special
design considerations.
\item Ethereum\cite{ethereum} is a public permissionless blockchain and
integrating complex procedures on its smart contracts is much trickier:
executed code consumes \textit{gas}, thus the general pursuit to minimize
the number of executed instructions. For complex procedures, oracles
are the best option (external data providers, not to be confused with
random oracles in cryptography or oracle machines in complexity theory).
The initial purpose of oracles is to provide data that didn't belong
in the blockchain (e.g., web pages) because decentralized applications
that achieve consensus shouldn't rely on off-chain sources of information.
But the mechanism can also be used for off-chain execution of complex
code. The steps to incorporate an oracle are:

\begin{enumerate}
\item Contract call to on-chain contract that executes complex computational
procedures (e.g., secure multi-party computation). Said call must
include enough \textit{gas} and the correct parameters:

\begin{enumerate}
\item Parameters should be encrypted to prevent that other participants
of the blockchain inspect them: a public key should be available for
this purpose; only the executing oracle should be able to decrypt
the parameters using the corresponding private key. Said executing
oracle must be implemented as a trusted server of the calling party;
otherwise, a more complex protocol involving outsourced oblivious
transfer must be used.
\item Contributed \textit{gas} is used to cover the costs of returning results
to the calling contract.
\end{enumerate}
\item The executing oracle receives the contract call, decrypts the data
and proceeds to the off-chain execution of the complex computational
procedure.
\item The executing oracle returns back the results to the calling contract
address.
\item Calling contract obtains the results: if large results are expected
(e.g., some few kilobytes) it's better to store them on IPFS\cite{1407.3561}
to prevent \textit{gas} costs, and what would be returned is a pointer
to said results.
\end{enumerate}
\end{itemize}
Regarding Bitcoin\cite{bitcoin}, although it would be possible to
use oracles\cite{orisi} similar to the ones used for Ethereum, these
are hardly used due to a combination of high transaction fees, high
confirmation time and low transaction rate.

This work is going to be open-sourced supporting Cothority\cite{cothority},
ByzCoin\cite{byzcoin} and OmniLedger\cite{cryptoeprint:2017:406}.

\subsubsection{Alternative Protocols and Standards}

Although this paper is focused on blockchains, it could be adopted
to other financial standards such as:
\begin{itemize}
\item Financial Information eXchange\cite{fixml}: messaging standard for
trade communication in the equity markets, with presence in the foreign
exchange, fixed income and derivatives market. 
\item Financial Products Markup Language (FpML)\cite{fpml}: XML messaging
standard for the Over-The-Counter derivatives industry.
\end{itemize}

\subsection{Functionality and Protocol}

The functionalities and protocols of this sub-section and section
\ref{sec:Functionalities-and-Protocols} constitute an open framework
on which to instantiate different secure multi-party computation protocols,
thus benefiting from upcoming research advances.

In this sub-section we present our secure protocol for private smart
contracts, consisting of the following standard functionality:

\fbox{\begin{minipage}[t]{1\columnwidth}%
\textbf{Functionality 4.1: Secure computation of smart contracts}
\begin{itemize}
\item \textbf{Parties:} $E_{1},...,E_{N}$, set of nodes $N_{i}$ of a blockchain
$B$
\item \textbf{Inputs:} smart contract $SC$, private inputs from parties
$E_{1}:\overrightarrow{x},E_{2}:\overrightarrow{y},...,E_{N}:\overrightarrow{z}$
\item \textbf{The functionality:}

\begin{enumerate}
\item Secure computation of smart contract $SC$
\item Results are returned and/or saved on the blockchain $B$
\end{enumerate}
\item \textbf{Output:} results from the secure computation $E_{1}:\overrightarrow{r_{1}},E_{2}:\overrightarrow{r_{2}},...,E_{N}:\overrightarrow{r_{N}}$.\end{itemize}
\end{minipage}}

Functionality 4.1 (Secure computation of smart contracts) is implemented
by the following protocol:\\

\fbox{\begin{minipage}[t]{1\columnwidth}%
\textbf{Protocol 4.2: Realising Functionality 4.1 (Secure computation
of smart contracts)}
\begin{itemize}
\item \textbf{Parties:} $E_{1},...,E_{N}$, set of nodes $N_{i}$ of a blockchain
$B$
\item \textbf{Inputs:} smart contract $SC$, private inputs from parties
$E_{1}:\overrightarrow{x},E_{2}:\overrightarrow{y},...,E_{N}:\overrightarrow{z}$
\item \textbf{The protocol:}

\begin{enumerate}
\item Parties $E_{1},...,E_{N}$ proceed to execute the smart contract $SC$:

\begin{enumerate}
\item The private inputs $\overrightarrow{x},\overrightarrow{y},...,\overrightarrow{z}$
and non-private inputs are sent to nodes $N_{i}$: 

\begin{enumerate}
\item Technically, this could be a HTTPS/REST call or an oracle call to
executing nodes $N_{i}$. 
\item The number of executing nodes $N_{i}$ depends on the setting of the
blockchain $B$ (public/private permissionless/permissioned): that
is, it could range from every node of the blockchain $B$ to just
a subset of nodes under PBFT consensus.
\item On permissioned blockchains, each executing node $N_{i}$ should be
a server owned/operated by the corresponding calling party $E_{i}$
(i.e., the servers are assumed to be trusted and the adversary is
on the network); but if the nodes are being run on a public cloud
or any other server outside the full control of the corresponding
calling party $E_{i}$, then a protocol for outsourcing secure computations
must be used (see \secref{Appendix-outsource}).
\end{enumerate}
\item The executing nodes $N_{i}$ proceed to execute the smart contract
$SC$: the execution command from parties $E_{1},E_{2},...,E_{N}$
contains the preferred secure computation engine to be used (i.e.,
there could be multiple execution engines with different protocols
under various security assumptions; if the parties don't agree, the
secure computation will not be carried out).
\end{enumerate}
\item If consensus on the results of the previous computation is reached,
then said results $\overrightarrow{r_{1}},\overrightarrow{r_{2}},...,\overrightarrow{r_{N}}$
could be returned to executing parties $E_{1},E_{2},...,E_{N}$ and/or
written on the blockchain $B$.
\end{enumerate}
\item \textbf{Output:} results from the secure computation $E_{1}:\overrightarrow{r_{1}},E_{2}:\overrightarrow{r_{2}},...,E_{N}:\overrightarrow{r_{N}}$.\end{itemize}
\end{minipage}}\\

The security of the protocol is proved on \secref{Security-Analysis}.

\subsection{Experimental Results}

Table \ref{tab:Execution-times} summarizes the execution cost for
several example applications, chosen for their economic significance:
\begin{enumerate}
\item Millionaire's Problem (i.e., determining who's got the bigger number
without revealing anything else)
\item Second-price auction: sealed-bid auction not revealing the bids between
the participants and without an auctioneer. The highest bidder wins,
but the price paid is the second-highest bid\cite{JOFI:JOFI2789}.
\item European Exchange Options: valuation of an option (the right, but
not the obligation) using Margrabe's formula\cite{RePEc:bla:jfinan:v:33:y:1978:i:1:p:177-86}
to exchange one risky asset for another risky asset at the time of
maturity; this example is useful to hedge private portfolios of volatile
crypto-tokens of physical assets. Suppose two risky assets with prices
$S_{1}\left(t\right)$ and $S_{2}\left(t\right)$ at time $t$ and
each with a constant dividend yield $q_{i}$: we calculate the option
to exchange asset 2 for asset 1 has a payoff $max\left(0,S_{1}\left(t\right)-S_{2}\left(t\right)\right)$,
\begin{eqnarray*}
Price\, option & = & S_{1}e^{-q_{1}t}N\left(d_{1}\right)-S_{2}e^{-q_{2}t}N\left(d_{2}\right)\\
\sigma & = & \sqrt{\sigma_{1}^{2}+\sigma_{2}^{2}-2\rho\sigma_{1}\sigma_{2}}\\
d_{1} & = & \frac{\ln\left(\frac{S_{1}}{S_{2}}\right)+\left(q_{2}-q_{1}+\frac{\sigma^{2}}{2}\right)t}{\sigma\sqrt{t}}\\
d_{2} & = & d_{1}-\sigma\sqrt{t}
\end{eqnarray*}
where $\rho$ is the Pearson's correlation coefficient of the Brownian
motion of $S_{i}$ and $\sigma_{i}$ are the volatilies of $S_{i}$.
Private inputs are $\sigma_{i}$, $S_{i}$ and $q_{i}$.
\item Currency Exchange Options: valuation of an option (the right, but
not the obligation) using the model of Garman-Kohlhagen\cite{RePEc:eee:jimfin:v:2:y:1983:i:3:p:231-237}
to exchange one currency for another at a fixed price; this example
is useful to hedge private portfolios of volatile crypto-currencies.
Suppose two risky currencies with different interest rates but constant
exchange rate: we calculate the calls and puts with the following
equations, 
\begin{eqnarray*}
Call & = & S_{0}e^{-\rho t}N\left(d_{1}\right)-X^{-rt}N\left(d_{2}\right)\\
Put & = & Xe^{-rt}N\left(-d_{2}\right)-S_{0}e^{-\rho t}N\left(-d_{1}\right)\\
d_{1} & = & \frac{\ln\left(\frac{S_{0}}{X}\right)+\left(r-\rho+\frac{\sigma^{2}}{2}\right)t}{\sigma\sqrt{t}}\\
d_{2} & = & d_{1}-\sigma\sqrt{t}
\end{eqnarray*}
where $r$ is the continuously compounded domestic interest rate,
$\rho$ is the continuously compounded foreign interest rate, $S_{0}$
is the spot rate, $X$ is the strike price, $t$ is the time to maturity
and $\sigma$ is the foreign exchange rate volatility. Private inputs
are $r,\rho$ and $S_{0}$ (with constant exchange rate).
\item Crowdfunding smart contract: a simple crowdfunding smart contract
is considered, that checks if the minimum contribution target is reached
and then returns the raised amount, or 0 otherwise.\\
\begin{algorithm}[H]
int crowdfund(int inputX, int inputY) \{\\
~~int ret;\\
~~int sum = inputX + inputY;\\
~~int minimum = 1000;\\

~~if ( sum >= minimum ) ret = sum; else ret = 0;\\

~~return ret;\\
\}

\protect\caption{Crowdfunding smart contract}
\end{algorithm}

\item DAO-like Investment Fund: a simple emulation of an investment fund
is considered, that check if the minimum contribution target is reached
and then returns the principal compounded after a number of years.\\
\begin{algorithm}[H]
float daoInvestFund(int inputX, int inputY) \{\\
~~float ret;\\
~~int sum = inputX + inputY;\\
~~int minimum = 1000;\\

~~if ( sum >= minimum ) ret = sum {*} (1 + (0.04/4))\textasciicircum{}(4{*}5);
else ret = 0;\\

~~return ret;\\
\}

\protect\caption{DAO-like Investment Fund}
\end{algorithm}

\item Double auction for decentralized exchanges: the secrecy of the prices
and quantities is maintained during the off-chain secure computation
and the settlement is finally done on-chain. Secure computation is
very useful in these settings because it enables secret order books,
increasing liquidity and price discovery under specially designed
mechanisms\cite{EPRINT:Jutla15}.
\end{enumerate}
\begin{table}[H]
\centering{}%
\begin{tabular}{|c|c|c|c|c|}
\hline 
Example & AND Gates & Time A & Time B & Time C\tabularnewline
\hline 
\hline 
Millionaire (int) & 96 & 1 & 1 & 0.485\tabularnewline
\hline 
Second-price Auction (int) & 192 & 1.1 & 1 & 0.862\tabularnewline
\hline 
European Exchange Options (float) & 267507 & 810 & 1273.8 & 1185.49\tabularnewline
\hline 
Currency Call Options (float) & 323529 & 979.6 & 1540.6 & 957.77\tabularnewline
\hline 
Crowdfunding smart contract(int) & 128 & 1 & 1 & 0.458\tabularnewline
\hline 
DAO-like Investment Fund(int) & 2144 & 6.5 & 10.2 & 0.458\tabularnewline
\hline 
Double auction (int) & 567829 & 1419 & 2552.2 & 2137.2\tabularnewline
\hline 
\end{tabular}\protect\caption{\label{tab:Execution-times}Execution times for application experiments.
The times shown are the wall clock time in millisecond to complete
the secure multi-party computation (i.e., from the initial oblivious
transfers to the final reveling of the results) for various secure
multi-party computation engines (A: semi-honest; B:malicious security;
C: secret-sharing (Sharemind\cite{EPRINT:BogLauWil08} estimation)).}
\end{table}

\subsection{Modes of Interaction}

\begin{figure}[H]
\centering{}\includegraphics[scale=0.5]{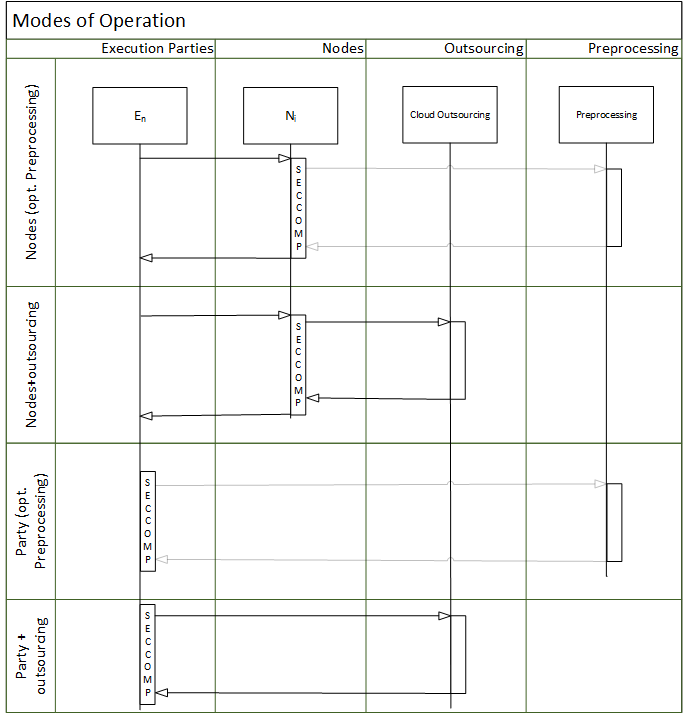}\protect\caption{Supported modes of interaction}
\end{figure}
Secure computation can be carried on the nodes of the blockchain or
on the parties themselves. Additionally, said secure computations
can be outsourced to the cloud (see \secref{outsourcing}) or use
mined pre-processing data for secure multi-party computation (see
\secref{miningPreprocess}).

\subsection{Outsourcing Secure Computations for Cloud-based Blockchains\label{sec:outsourcing}}

The following scheme makes use of a multi-party non-interactive key
exchange\cite{EPRINT:BonZha13b,cryptoeprint:2017:152} to establish
a shared secret between the computing parties and the executing nodes
of the blockchain: replacing the NIKE protocol would require a PKI
infrastructure between the computing parties and the executing nodes
that will be used to establish shared secrets between them; it's viable
but also more cumbersome within the context of public permissionless
ledgers. A multiparty non-interactive key exchange (NIKE) scheme consists
of the following algorithms:
\begin{itemize}
\item $Setup(M,N,\lambda)$: this algorithm outputs public parameters $params$,
taking as input $M$, the maximum number of parties that can derive
a shared key, $N$, the maximum number of parties in the scheme, and
the security parameter $\lambda$.
\item $Publish(params,\, i)$: this algorithm outputs the party's secret
key $sk_{i}$, and a public key $pk_{i}$ which the party publishes;
inputs are public parameters $params$ and the party's index $i$.
\item $KeyGen(params,i,sk_{i},S,\{pk_{j}\}_{j\in S})$: this algorithm outputs
a shared key $k_{S}$; inputs are public parameters $params$, the
party's index $i$, the set $S\subseteq\left[N\right]$ of size at
most $M$ and the set of public keys $\{pk_{j}\}_{j\in S}$ of the
parties in $S$.
\end{itemize}
As specified below, the functionality for outsourcing secure computations
satisfies the following requirements: 
\begin{itemize}
\item \textit{offline parties}: no parties need to be involved when the
outsourced secure computation is executed.
\item \textit{private parameters reuse}: parties don't need to re-upload
their private parameters.
\item \textit{restricted collusion}: an adversary may corrupt a subset of
the parties or the nodes of the blockchain, but not both\cite{EPRINT:KamMohRay11}.\\

\end{itemize}
\fbox{\begin{minipage}[t]{1\columnwidth}%
\textbf{Functionality B.1: Outsourcing for Cloud-based Blockchains}
\begin{itemize}
\item \textbf{Parties:} $E_{1},E_{2},...,E_{N}$, set of nodes $N_{i}$
of a blockchain $B$ ($N_{G}$ being the garbling node and $N_{E}$
being the evaluator node)
\item \textbf{Inputs:} smart contract $SC$, private inputs from parties
$E_{1}:\overrightarrow{x_{i}},E_{2}:\overrightarrow{x_{i}},...,\, E_{N}:\overrightarrow{x_{i}}$
\item \textbf{The functionality:}

\begin{enumerate}
\item $KeyExchange$: run a multi-party NIKE between parties $E_{1},E_{2},...,E_{N}$
and nodes $N_{i}$ and publish their public keys $pk_{i}$; from said
public keys, a secret key $k_{S}$ is derived.
\item $SendPrivateParameters(E_{i})$: send $\left(SendPrivateParameters,\, E_{i}\right)$
to $N_{E}$.
\item $SECCOMP\left(SC\right)$: upon the reception of a computation request
from party $E_{k}$ for a function in smart contract $SC$, compute
$y=SC\left(x_{k},\left\{ x_{j}|\forall j\in N:E_{j}\right\} \right)$;
send $\left(result,y\right)$ to $E_{k}$ and $\left(SECCOMP,SC,E_{k}\right)$
to $N_{E}$.\end{enumerate}
\end{itemize}
\end{minipage}}\\

Let $Enc$ be a symmetric-key encryption algorithm secure against
Chosen-Plaintext Attacks. The following protocol uses a pivot table
during the secure computation, allowing the evaluator node to obliviously
map the encoded inputs by the parties to the encoding expected by
the circuit created by the garbling node. To implement Functionality
B.1 (Outsourcing for Cloud-based Blockchains), the following protocol
is proposed: \\

\fbox{\begin{minipage}[t]{1\columnwidth}%
\textbf{Protocol B.2: Realising Functionality B.1 (Outsourcing for
Cloud-based Blockchains)}
\begin{itemize}
\item \textbf{Parties:} $E_{1},E_{2},...,E_{N}$, set of nodes $N_{i}$
of a blockchain $B$ ($N_{G}$ being the garbling node and $N_{E}$
being the evaluator node)
\item \textbf{Inputs:} smart contract $SC$, private inputs from parties
$E_{1}:\overrightarrow{x_{i}},E_{2}:\overrightarrow{x_{i}},...,E_{N}:\overrightarrow{x_{i}}$
\item \textbf{The protocol:}

\begin{enumerate}
\item $KeyExchange$: run a multi-party NIKE between parties $E_{1},E_{2},...,E_{N}$
and nodes $N_{i}$ and publish their public keys $pk_{i}$; from said
public keys, a secret key $k_{S}$ is derived.

\begin{enumerate}
\item The setup phase is executed: $params:=Setup(M,N,\lambda)$
\item Each party $i$ runs $pk_{i},sk_{i}:=Publish(params,\, i)$ and publish
$pk_{i}$
\item Each party $i$ runs $k_{S}:=KeyGen(params,i,sk_{i},S,\{pk_{j}\}_{j\in S})$
to obtain the shared secret key
\end{enumerate}
\item $SendPrivateParameters(E_{i}:\overrightarrow{x_{i}},k_{S})$: party
$E_{i}$ chooses nonce $n_{i}$ and for every bit of $\overrightarrow{x_{i}}$
computes $X_{il}^{\overrightarrow{x_{i}}\left[l\right]}=PRF_{k_{S}}\left(\overrightarrow{x_{i}},l,n_{i}\right)$;
then, party $E_{i}$ sends these to $N_{E}$ and also sends $n_{i}$.
$N_{E}$ stores $\left(\left(X_{i1}^{x_{i}\left[1\right]},...,X_{il}^{x_{i}\left[l\right]}\right),n_{i}\right)$.
\item $SECCOMP\left(SC\right)$:

\begin{enumerate}
\item Secure computation at node $N_{G}$: the garbling node $N_{G}$ compiles
and garbles smart contract $SC$ into the garbled circuit $GC_{SC}$.
Then, for each party $E_{j}$ and index $l$ of the length of $\overrightarrow{x_{i}}$:

\begin{enumerate}
\item For each party $E_{j}$, retrieve nonce $n_{j}$ from $N_{e}$.
\item Compute pivot keys: generate $s_{jl}^{0}=PRF_{k_{S}}\left(0,l,n_{j}\right),\, s_{jl}^{1}=PRF_{k_{S}}\left(1,l,n_{j}\right)$.
\item Compute and save garbled inputs: using the pivot keys, encrypt $Enc_{s_{jl}^{0}}\left(w_{jl}^{0}\right)$
and $Enc_{s_{jl}^{1}}\left(w_{jl}^{1}\right)$; then save them into
pivot table $P_{q}\left[j,l\right]$ in random order.
\end{enumerate}
\item Secure computation at node $N_{E}$: for every bit of $x_{j}$ and
using the encoding $X_{jl}^{x_{j}\left[l\right]}$, decrypt the correct
garbled values of each $E_{j}$ from $P_{q}$; evaluate the garbled
circuit $GC_{SC}$ and send the output to $N_{G}$.
\item Result at $N_{G}$: decode the output to obtain the results $\overrightarrow{r_{1}},\overrightarrow{r_{2}},...,\overrightarrow{r_{N}}$
. 
\end{enumerate}
\end{enumerate}
\item \textbf{Output: }results from the secure computation $E_{1}:\overrightarrow{r_{1}},E_{2}:\overrightarrow{r_{2}},...,E_{N}:\overrightarrow{r_{N}}$.\end{itemize}
\end{minipage}}\\

\begin{thm}
(Outsourcing for Cloud-Based Blockchains). Assuming secure channels
between the parties $E_{1},E_{2},...,E_{N}$ and the nodes $N_{i}$
of the blockchain $B$, Protocol B.2 securely realises Functionality
B.1 against static corruptions in the semi-honest security model with
the garbling scheme satisfying privacy, obliviousness and correctness.\end{thm}
\begin{proof}
See Appendix \ref{sec:Appendix-outsource}.
\end{proof}

\subsection{Mining pre-processing data for Secure Multi-Party Computation\label{sec:miningPreprocess}}

The Proof-of-Work of crypto-currencies consumes great amounts of computational
power and electricity: 14TWh for Bitcoin\cite{bitcoinConsumption}
and 4.25TWh for Ethereum\cite{ethereumConsumption} just calculating
hash functions. Miners could create pre-processing data for secure
multi-party computation and be incentivised with crypto-tokens: 50-80\%
of total execution time is spent on pre-processing depending on the
function/protocol, thus they would be profiting from ``renting''
their computational power to save significant amounts of computational
time.

A recent paper considers the case of outsourcing MPC-Preprocessing\cite{cryptoeprint:2017:262}
to third parties and then computing parties reusing the pre-processed
data (i.e., SPDZ-style authenticated shares) in the online phase:
it's especially efficient if there is a subset of parties trusted
by all the computing parties which can do all of the pre-processing
and then distribute it to the computing parties. It's based on the
re-sharing technique of \cite{BGW88}, but without using zero-knowledge
proofs: it also fits into another recent protocol for secure multi-party
computation\cite{cryptoeprint:2017:189,emp-toolkit} that is much
more efficient for WAN networks than SPDZ derivatives, except that
\cite{cryptoeprint:2017:189} uses BDOZ-style authenticated shares
instead of SPDZ-style authenticated shares:
\begin{itemize}
\item BDOZ-style\cite{EPRINT:BDOZ10} authenticated shares: for each secret
bit $x$, each party holds a share of $x$; for each ordered pair
of parties $\left(P_{i},P_{j}\right)$, $P_{i}$ authenticates its
own share to $P_{j}$. Specifically, when party $P_{i}$ holds a bit
$x$ authenticated by $P_{j}$, this means that $P_{j}$ is given
a random key $K_{j}\left[x\right]\in\left\{ 0,1\right\} ^{k}$ and
$P_{i}$ is given the MAC tag $M_{j}\left[x\right]:=K_{j}\left[x\right]\oplus x\triangle_{j}$,
where $\triangle_{i}\in\left\{ 0,1\right\} ^{k}$ is a global MAC
key held by each party. Let $\left[x\right]^{i}$ denote an authenticated
bit where the value of $x$ is known to $P_{i}$ and is authenticated
to all other parties: that is, $\left(x,\left\{ M_{k}\left[x\right]\right\} _{k\neq i}\right)$
is given to $P_{i}$ and $K_{j}\left[x\right]$ is given to $P_{j}$
for $j\neq i$. An authenticated shared bit $x$ is generated by XOR-sharing
$x$ and then distributing the authenticated bits $\left\{ \left[x^{i}\right]^{i}\right\} $:
let $\left\langle x\right\rangle :=\left(x^{i},\left\{ M_{j}\left[x^{i}\right],K_{i}\left[x^{j}\right]\right\} _{j\neq i}\right)$
denote the collection of these authenticated shares for $x$.
\item SPDZ-style\cite{EPRINT:DPSZ11} authenticated shares: each party holds
a share of a global MAC key; for a secret bit $x$, each party holds
a share $x$ and a share of the MAC on $x$. Specifically, a value
$x\in\mathbb{\mathbb{F}}_{q}$ is secret shared among parties $P$
by sampling $\left(x_{i}\right)_{i\in P}\leftarrow\mathbb{F}_{q}^{\left|P\right|}$
subject to $x=\sum_{i\in P}x_{i}$ with a party $i$ holding the value
$x_{i}$; the MAC is obtained by sampling $\left(\gamma\left(x\right)_{i}\right)_{i\in P}\leftarrow\mathbb{F}_{q}^{\left|P\right|}$
subject to $\sum_{i\in P}\gamma\left(x\right)_{i}=\alpha\cdot x$
and party $i$ holding the share $\gamma\left(x\right)_{i}$: let
$\left\langle x\right\rangle :=\left(\left(x_{i}\right)_{i\in P},\left(\gamma\left(x\right)_{i}\right)_{i\in P}\right)$
to denote that $x$ is an authenticated secret share value, where
party $i\in P$ holds $x_{i}$ and $\gamma\left(x\right)_{i}$, under
a global MAC key $\alpha=\sum_{i\in P}\alpha_{i}$.
\end{itemize}
It's straightforward to adapt the protocol $\varPi_{Prep}^{R\rightarrow Q,\overline{A}}$
from \cite{cryptoeprint:2017:262} to process BDOZ-style authenticated
shares instead of SPDZ-style authenticated shares: thus, the authenticated
shares generated by the offline pre-processing parties would be reshared
amongst the computing parties as required before executing any secure
multi-party computation and the pre-processing parties would be incentivised
with crypto-tokens. Another recent paper\cite{poolGC} concurrently
considered the generation of a pool of garbled gates to provide on-demand
secure computation services between two parties, a similar concept
except that it doesn't include outsourcing from third parties.

\subsubsection{Security setting}

Let $E$ denote the set of $n_{E}$ parties who are to run the online
phase and $O$ the set of $n_{O}$ outsourcing parties that run the
pre-preprocessing for the executing parties $E$ (respectively $Q$
and $R$ in \cite{cryptoeprint:2017:262}). Adversaries can corrupt
a majority of parties in $E$ and in $O$, but not all parties in
$E$ nor all parties in $O$: that is, each honest party in $E$ believes
that there is at least one honest party in $O$, but they may not
know which one is honest. Let $t_{E}$ denote the number of parties
in $E$ that are corrupt (resp. $t_{O}$ in $O$): the associated
ratios are denoted by $\epsilon_{E}=t_{E}/n_{E}$ and $\epsilon_{O}=t_{O}/n_{O}$.
The executing parties $E$ are divided into subsets $\left\{ E_{i}\right\} _{i\in O}$
forming a cover, with a party in $O$ associated with each subset:
a cover is defined to be secure if at least one honest party in $O$
is associated to one honest party in $E$.

In public permissionless blockchains, it's expected that there is
no prior trust relation between parties in $E$ and $O$: an efficient
algorithm is offered in \cite{cryptoeprint:2017:262} for assigning
a cover to the network of parties so that the adversary can only win
with negligible probability in the security parameter $\lambda$ in
the case where the covers are randomly assigned, and working out the
associated probability of obtaining a secure cover. It assumes that
each party in $O$ sends to the same number of parties $l\geq\left\lceil n_{E}/n_{O}\right\rceil $
in $E$. The high-level idea of the algorithm is the following:
\begin{enumerate}
\item For each party in $E$, we assign a random party in $O$, until each
party in $O$ has $\left\lceil n_{E}/n_{O}\right\rceil $ parties
in $O$ assigned to it
\item For each party in $O$, we assign random parties in $E$ until each
party in $O$ has $l$ total parties which it sends to.
\end{enumerate}
The probability to obtain a secure cover is given by 
\[
1-\frac{t_{E}!\cdot\left(n_{E}-\left(n_{O}-t_{O}-1\right)\left\lceil n_{E}/n_{O}\right\rceil \right)!}{n_{E}!\cdot\left(t_{E}-\left(n_{O}-t_{O}-1\right)\left\lceil n_{E}/n_{O}\right\rceil \right)!}\cdot\left(\frac{\left(\begin{array}{c}
t_{E}-\left\lceil n_{E}/n_{O}\right\rceil \\
l-\left\lceil n_{E}/n_{O}\right\rceil 
\end{array}\right)}{\left(\begin{array}{c}
n_{E}-\left\lceil n_{E}/n_{O}\right\rceil \\
l-\left\lceil n_{E}/n_{O}\right\rceil 
\end{array}\right)}\right)^{n_{O}-t_{O}-1}
\]

In case where all but one party is corrupt in each $E$ and $O$,
then the probability to obtain a secure cover is given by $l/n_{E}$.

\subsubsection{Preventing Sybil attacks}

In the context of blockchains, Sybil attacks in which an attacker
creates a large number of pseudonymous identities can easily be prevented:
before running secure computations, any party could be required to
deposit some arbitrarily high amount of money on a smart contract
that would be confiscated in case any abnormal behaviour is detected
(e.g., selective failure attacks, aborts, ...). Note that this simple
technique maintains the anonymity of the parties and prevents that
the computed pre-processing data gets intentionally wasted.

\subsection{Other Applications}

Applications of this technology can be found on many commercial/financial
settings. Some of the most noteworthy are as follows:
\begin{itemize}
\item the most immediate application of secure multi-party computation is
the removal of third parties and the financial industry has plenty
of them: market makers, escrows, custodians, brokers, even banks themselves
are intermediaries between savers and borrowers.
\item crypto-banks: all the financial information contained within bank's
databases could be encrypted with secure computation techniques. Users
could make deposits, take loans and trade financial instruments without
ever revealing their financial positions/transactions to inquiring
third-parties.
\end{itemize}
Other applications mentioned in literature:
\begin{itemize}
\item economic/financial applications: financial exposures could be shared
between mutually distrusting parties without revealing any confidential
information to better control financial risks\cite{shareRiskExposure};
credit scoring\cite{EPRINT:DDNNT15,epub27630}; facilitate the work
of financial supervisors\cite{katzFederalReserve} without compromising
confidentiality; protect the privacy of data in online marketplaces\cite{cryptoeprint:2011:257};
re-implement the stock market \cite{EPRINT:Jutla15} without a trusted
auctioneer and no party learning the order book; remove escrows\cite{kumaresanPrivacyContracts}
with claim-or-refund transactions and secure computation.
\item game theory/mechanism design: remove trusted third-parties in tâtonnement
algorithms\cite{tatonnement} for one-time markets, allowing to privately
share the utility function of the involved parties without revealing
it and obtaining an incentive compatible protocol in the process;
privacy-preserving auctions\cite{cryptoeprint:2016:797,cryptoeprint:2017:439};
more generally, implement mechanisms respecting privacy to obtain
incentive-compatibility (e.g., auctions\cite{Naor99privacypreserving,EPRINT:ElkLip03})
or that incentivise data-driven collaboration among competing parties\cite{EPRINT:AzaGolPar15}.
\item statistics: benchmark the performance of companies within the same
sector\cite{EPRINT:BogTalWil11}; establish correlation and causation\cite{EPRINT:BKKRST15}
between confidential datasets.
\end{itemize}

\section{Verifiable Smart Contracts\label{sec:Verifiable-Smart-Contracts}}

After the DAO attack, there has been some work (\cite{formalEVM,Bhargavan:2016:FVS:2993600.2993611,coq-smartcontracts,idrisEVM,kevm})
to include formal methods in the development of smart contracts: unfortunately,
current solutions are very complex and cumbersome\cite{formalEVM,Bhargavan:2016:FVS:2993600.2993611,idrisEVM,kevm},
almost equivalent to formally verifying assembly code. Only very high-level
languages should be used to write smart contracts, not assembly-like
ones (EVM): even the C language should be considered too low-level
for these purposes because the required proofs must contain all kind
of details about memory management and pointers.

Not all smart contracts need to be formally verified, and not even
every part of their code should be. On permissioned blockchains, non-verified
smart contracts will be much more accepted than on public permissionless
blockchains.

Unlike other works that only consider the correctness of the computed
output and resort to resource-consuming zero-knowledge proofs \cite{EPRINT:BCGGMT14},
the mathematical proofs considered here are multi-purpose: invariants,
pre- and post-conditions, termination, correctness, security, resource
consumption, legal/regulatory (e.g., self-enforcement), economic (e.g.,
fairness, double-entry consistency, equity), functional and any other
desirable property that can be mathematically expressed.

\begin{figure}[H]
\centering{}\includegraphics[scale=0.5]{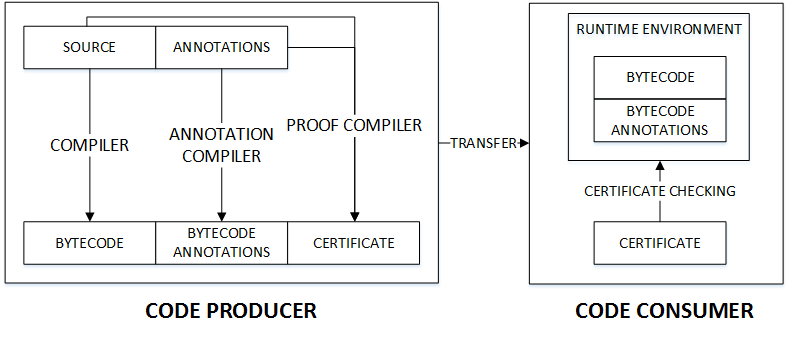}\protect\caption{Proof-Carrying Code infrastructure}
\end{figure}

The solution of how to execute untrusted code from potentially malicious
sources is not new, said problem was already considered for mobile
code: Proof-Carrying Code\cite{Necula:1998:SUA:648051.746192} prescribes
accompanying the untrusted code with proofs (certifiable certificates)
that can be checked before execution to verify their validity and
compare the conclusions of the proofs to the security policy of the
code consumer to determine whether the untrusted code is safe to execute.
On the positive side, no trust is required on the code producer and
there is no runtime overhead during execution, but the code must be
annotated (e.g., see pseudo-code annotation of the Program Listing
below) and detailed proofs generated, a process which can be complex
and costly. A novel variant combining PCC with Non-Interactive Zero-Knowledge
proofs is proposed in the next sub-section \ref{sec:ZK-Proofs}.

\begin{algorithm}[H]
class Account \{

~~int balance; // \textbf{\textit{\footnotesize{}invariant}} balance
>= 0;

~~// \textbf{\textit{\footnotesize{}requires}} amt >= 0 

~~// \textbf{\textit{\footnotesize{}ensures}} balance == amount 

~~Account(int amount) \{ balance = amount; \}\\

~~// \textbf{\textit{\footnotesize{}ensures}} balance == acc.balance 

~~Account(Account \_account) \{ balance = \_account.balance(); \}\\

~~// \textbf{\textit{\footnotesize{}requires}} amount > 0 \&\& amount
<= \_account.balance() 

~~// \textbf{\textit{\footnotesize{}ensures}} balance == \textbackslash{}old(balance)
+ amount 

~~//   \&\& \_account.balance == \textbackslash{}old(\_account.balance
- amount); 

~~transfer(int amount, Account \_account) 

~~\{ \_account.withdraw(amount); deposit(amount); \}\\

~~// \textbf{\textit{\footnotesize{}requires}} amount > 0 \&\& amount
<= balance 

~~// \textbf{\textit{\footnotesize{}ensures}} balance == \textbackslash{}old(balance)
- amount 

~~void withdraw(int amount) \{ balance -= amount; \}\\

~~// \textbf{\textit{\footnotesize{}requires}} amount > 0; 

~~// \textbf{\textit{\footnotesize{}ensures}} balance == \textbackslash{}old(balance)
+ amount 

~~void deposit(int amount) \{ balance += amount; \}\\

~~// \textbf{\textit{\footnotesize{}ensures}} \textbackslash{}result
== balance 

~~int balance() \{ return balance; \} 

\}

\protect\caption{Example pseudo-code with annotations}

\end{algorithm}

The verification of a program typically requires many annotations,
at least a 1:1 ratio of lines of code against specifications: tools
have been developed to automatically generate said annotations and
are very useful in this setting. In case of using an interactive theorem
prover, an extensive library of tactics must help to handle complex
cases. Smart contracts must be written with the specific purpose of
verification in mind, otherwise it becomes extremely complex to generate
complete proofs\cite{Woodcock:2007:CME:1341231.1341236}; regarding
the size, effort and duration of the verification process, there is
a strong linear relationship between effort and proof size\cite{Staples:2014:PPE:2652524.2652551}
and a quadratic relationship between the size of the formal statement
and the final size of its formal proof\cite{Matichuk:2015:EST:2818754.2818842}.

The most cost-effective option in some settings (e.g., permissioned
blockchains) could be the publication of the annotated smart contracts
to the blockchain and not using some of the more advanced options
like the PCC toolchain or any kind of theorem prover, reversing the
burden of proof to code consumers but in some sense helping them with
the provided annotations. Therefore, there is a scale of Verification
Levels when publishing smart contracts on blockchains:
\begin{enumerate}
\item Annotated smart contracts
\item Annotated smart contracts automatically tested using heuristics/concolic
execution
\item Annotated smart contracts with full/partial proofs
\item Annotated smart contracts with certifiable certificates (Proof Carrying
Code)
\end{enumerate}
Ultimately, as an example of their applicability, the use of the proposed
annotated smart contracts in combination with the zero-knowledge proofs
of \secref{ZK-Proofs} allows for an alternative way to implement
the proofs of assets, liabilities and solvency of exchanges of \cite{EPRINT:DBBCB15}.

\subsection{Case Study}

\begin{algorithm}[H]
class Crowdfunding \{

~~~~int minimum = 1000;

~~~~// requires 0 < n

~~~~// ensures \textbackslash{}result >= minimum

~~~~public int crowdfund(int n, int{[}{]} inputs) \{

~~~~~~int sum = 0;

~~~~~~// invariant 0 <= i \&\& i <= n

~~~~~~for (int i = 0; i < n; i++) \{

~~~~~~~~sum += inputs{[}i{]};

~~~~~~\}

~~~~~~return sum;

~~~~\} 

\}\protect\caption{Crowdfund example processed with PCC toolchain}

\end{algorithm}

Proofs are written in Coq: the following execution statistics of the
PCC toolchain are reported,
\begin{itemize}
\item Proof-generation time overhead: directly correlated to the size of
the compiled program being analysed. It follows approximately the
following formula on a modern laptop (Intel\textregistered{} Core™
i7-7500U 2.7Ghz): 
\end{itemize}
\[
time(bytecode\, size\, bs)=1.5+\left(\frac{bs}{1500}\right)\, secs
\]

\begin{itemize}
\item Proof-verification time overhead: (less) correlated to the size of
the compiled program being analysed. It follows approximately the
following formula on a modern laptop (Intel\textregistered{} Core™
i7-7500U 2.7Ghz):
\end{itemize}
\[
time(bytecode\, size\, bs)=0.25+\left(\frac{bs}{6000}\right)secs
\]

\begin{itemize}
\item Certificate size overhead: 30\%
\end{itemize}
\begin{figure}[H]
\centering{}\includegraphics{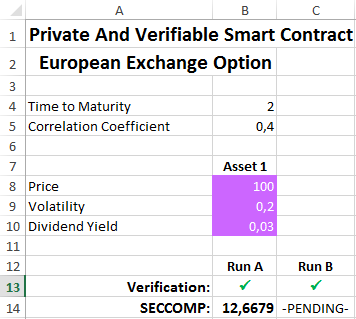}\protect\caption{Secure spreadsheet\cite{seccompSecureSpreadsheet} enabled for privacy-preserving
computation displaying the result of private and verifiable smart
contracts (e.g., cryptographically secure financial instruments and
their derivatives)}
\end{figure}

Not much thought has been given to the rather practical question of
what user interface should smart contracts use: calculations will
be done on the returned values of smart contracts, even encrypted
ones; and nested calculations on their inputs/outputs is also required.
A spreadsheet enabled for secure computation fits all the given requirements,
especially given that it's well accepted on the financial industry
and many are trained on its use.

\subsection{On the Size of Certifiable Certificates}

Since the invention of Proof Carrying Code\cite{Necula:1997:PC:263699.263712},
the question of how to efficiently represent and validate proofs\cite{Necula:1998:ERV:788020.788923}
has been a focus of much research. Different techniques are considered
in this work to reduce the size of certificates:
\begin{enumerate}
\item The reflection technique\cite{Allen90thesemantics}. This technique
decreases the size of the proof using the reduction offered by a proof
assistant: multiple explicit rewriting steps are replaced by implicit
reductions. In other words, implicit computation replaces explicit
deduction, achieving a reduction of proof size of orders of magnitude.
\item Deep embedding of the verification condition generator reduces the
size of proof terms due to the use of reflective tactics with the
subsequent reduction in the time needed to check the proofs.
\item Using hybrid methods that include static analysis: using the information
provided by type systems to eliminate unreachable paths reduces the
number of proof obligations. For example, null-pointer analysis is
used to prove that most accesses are safe and this information is
then employed to reduce the number of proof obligations.
\item Partitioning the certificate to reduce the memory required to check
it by running an interactive protocol between the code consumer and
producer.
\item Abstract interpretation is used to verify safety policies that reduce
the control flow graph. A fixed-point abstract model accompanies the
code whose validity implies compliance with safety policies: then,
the code consumer checks its validity in a single pass. Note that
there’s an inherent tradeoff between certificate size and checking
time: in order to reduce the size of the certificate, the code consumer
could generate the fixpoint but that would also increase the checking
time. Thus, it’s important to send the smallest subset of the abstract
model while maintaining an efficient single-pass check: that is, to
only store the information that the code consumer is not able to reproduce
by itself.
\item The size of certificates using proofs in sequent calculus is reduced
since it can be reconstructed from the used lemmas (i.e., using the
cut-inference rule). Proofs are represented in tables that a proof
engine will check\cite{DBLP:journals/corr/abs-0708-2252}: in order
to reduce the size of the table, some atomic formulas can be omitted
if the code consumer will be able to easily reprove them. 
\end{enumerate}
Other approaches that could be considered in order to reduce the size
of the proof include transmitting a proof generator\cite{2010LNCS.6480.249P}
to the code consumer instead of the proof itself, and executing the
proof generator on the consumer side to re-generate the proof using
a virtual machine.

\subsection{Zero-Knowledge Proofs of Proofs\label{sec:ZK-Proofs}}

Although it's possible to obfuscate certifiable certificates in such
a way that de-obfuscation wouldn't be any easier\cite{codeObfuscation}
while keeping the certifiable certificate sound and complete, this
level of security isn't acceptable in a formal cryptographic model.
When smart contracts are fully encrypted with homomorphic encryption/IO,
certifiable certificates reveal too much information about the code
and additional cryptographic protection is absolutely necessary: it's
possible to generate zero-knowledge proofs of proofs\cite{Blum87howto}
(or certifiable certificates, which are shorter), in such a way that
the code producer can convince the code consumer of the existence
and validity of proofs about the code without revealing any actual
information about the proofs themselves or the code of the smart contract.

Classically, this would require to come by an interactive proof system\cite{gmr85}
were the code consumer is convinced, with overwhelming probability,
of the existence and validity of proofs of the code through interaction
with a code producer; then, a zero-knowledge proof system will be
obtained using the methods of \cite{BGG90,directMKC}: for a more
detailed description on how to prove a theorem in zero-knowledge,
see \cite{Pope04provinga}. Lately, advances in verifiable computation
have produced advanced zero-knowledge proof systems to prove correctness
of remote execution: although many of these results could theoretically
be extended to prove general assertions about the code, the slowdowns
for proving would be higher than the current $10^{5}-10^{7}$ for
the very optimized case of proving correctness of executions\cite{Walfish:2015:VCW:2728770.2641562}.

Until the advent of methods to obtain zero-knowledge proofs from garbled
circuits\cite{EPRINT:JawKerOrl13} (i.e., general purpose ZK), zero-knowledge
proofs have been difficult to come by. \textit{ZKBoo}\cite{EPRINT:GiaMadOrl16},
a later development of \cite{EPRINT:JawKerOrl13}, creates non-interactive
zero-knowledge proofs for boolean circuits: this line of work is the
preferred choice to obtain NIZK-proofs of the validity of certifiable
certificates of smart contracts, because zk-SNARKs\cite{EPRINT:BCCT12}
would be very succinct, but with setup assumptions and much slower
to proof (e.g., the circuit \textit{\footnotesize{}C\_POUR} from Zerocash\cite{EPRINT:BCGGMT14}
has 4.109.330 gates and a reported execution time of 2 min). Lately,
\textit{ZKB++\cite{cryptoeprint:2017:279}} provides proofs that are
less than half the size than \textit{ZKBoo} and \textit{Ligero\cite{ligeroArguments}}
four time shorter than \textit{ZKB++.}

Using \textit{ZKBoo}, the statement to be proved would be “I know
a certifiable certificate such that $\phi\left(certificate\right)=true$”
for a certificate validation circuit $\phi$ and $L_{\phi}$ the language
$\left\{ true|\exists certificate\, s.t.\,\phi\left(certificate\right)=true\right\} $
with soundness error $2^{-80}$: minimizing the certificate validation
circuit $\phi$ as much as possible would be the most important optimization.

\paragraph*{Proof-Carrying Data}

A conceptually related technique to Proof-Carrying Code is Proof-Carrying
Data\cite{ct10}: in a distributed computation setting, PCD allows
messages to be accompanied by proofs that said messages and the history
leading to them follow a compliance predicate, in such a way that
verifiers can be convinced that the compliance predicate held throughout
the computation, even in the presence of malicious parties. PCC and
PCD are complimentary since PCD can enforce properties expressed via
PCC: interestingly, PCD could enable zero-knowledge privacy for PCC\cite{EPRINT:ChoTroVau13},
but at a greater efficiency cost.

\section{Private and Verifiable Smart Contracts}

In order to merge the Private smart contracts (see Section \ref{sec:Private-Smart-Contracts})
and the Verifiable smarts contracts (see Section \ref{sec:Verifiable-Smart-Contracts}),
the core of the present work has been re-developed and this will be
the version that is going to be open-sourced:
\begin{itemize}
\item private smart contracts are developed in Obliv-Java, a modified Open-JDK
for secure computation analogous to Obliv-C\cite{cryptoeprint:2015:1153}:
variables are annotated with ``\textit{@Obliv(partyId)}'' and all
the secure computation is executed transparently to the developer.
In comparison with ObliVM\cite{oblivm}, it's faster, more secure
and faithfully respects the original Java language standard and bytecode
semantics, thus it's fully compatible with all the Java ecosystem
and libraries.
\item private smart contract can be annotated with JML annotations for verification
purposes, and said verifications proved and the results discharged
in the form of proof-carrying code.
\end{itemize}

\section{Functionalities and Protocols\label{sec:Functionalities-and-Protocols}}

\subsection{General Overview}

\begin{figure}[H]
\centering{}\includegraphics[scale=0.65]{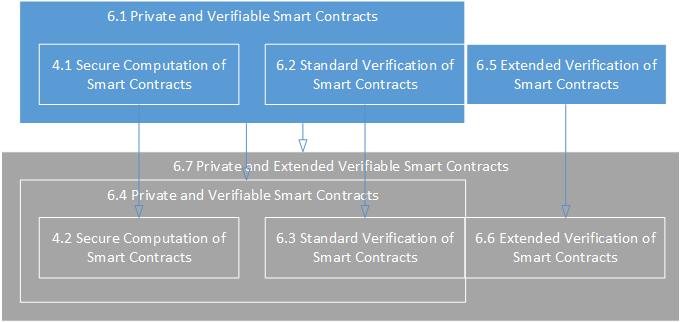}\protect\caption{General overview of the functionalities (top) and realised protocols
(bottom) described in this paper}
\end{figure}

\subsection{Detailed Description}

In this section we present our secure protocol for private and verifiable
smart contracts, specified in the Functionality below:

\fbox{\begin{minipage}[t]{1\columnwidth}%
\textbf{Functionality 6.1: Private and Verifiable Smart Contracts}
\begin{itemize}
\item \textbf{Parties:} $E_{1},...,E_{N}$, set of nodes $N_{i}$ of a blockchain
$B$
\item \textbf{Inputs:} smart contract $SC$, private inputs from parties
$E_{1}:\overrightarrow{x},E_{2}:\overrightarrow{y},...,E_{N}:\overrightarrow{z}$
\item \textbf{The functionality:}

\begin{enumerate}
\item Parties $E_{1},...,E_{N}$ verify the smart contract $SC$
\item Parties $E_{1},...,E_{N}$ securely execute the smart contract $SC$
on nodes $N_{i}$
\end{enumerate}
\item \textbf{Outputs:} results from the secure computation $E_{1}:\overrightarrow{r_{1}},E_{2}:\overrightarrow{r_{2}},...,E_{N}:\overrightarrow{r_{N}}$.\end{itemize}
\end{minipage}}\\

Our protocols for realising Functionality 6.1 consists of a standard
verification functionality and a secure computation functionality
(Functionality 4.1):\\

\fbox{\begin{minipage}[t]{1\columnwidth}%
\textbf{Functionality 6.2: Standard verification of smart contracts}
\begin{itemize}
\item \textbf{Parties:} $E_{1},...,E_{N}$
\item \textbf{Inputs:} smart contract $SC$
\item \textbf{The functionality:}

\begin{enumerate}
\item Parties $E_{1},...,E_{N}$ obtain the annotations/proofs/certificates
of smart contract $SC$
\item Verify the annotations, proofs and/or certificates available on the
smart contract $SC$
\item Check if the verified annotations, proofs and/or certificates are
sufficient for the smart contract $SC$ to be declared correct, according
to the local policies of each party
\end{enumerate}
\item \textbf{Output:} \textit{true} if the verification was correct, \textit{false}
otherwise.\end{itemize}
\end{minipage}}

To implement Functionality 6.2 (Standard verification of smart contracts),
the following protocol is used:\\

\fbox{\begin{minipage}[t]{1\columnwidth}%
\textbf{Protocol 6.3: Realising Functionality 6.2 (Standard verification
of smart contracts)}
\begin{itemize}
\item \textbf{Parties:} $E_{1},...,E_{N}$
\item \textbf{Inputs:} smart contract $SC$
\item \textbf{The protocol:}

\begin{enumerate}
\item Parties $E_{1},...,E_{N}$ download the annotations/proofs/certificates
of smart contract $SC$ and check their digital signatures.

\begin{enumerate}
\item If annotated smart contracts do not contain any kind of proofs or
certificates, the annotations of the smart contract are automatically
tested using an heuristic/concolic execution engine.
\item If annotated smart contracts contain full/partial proofs, they are
locally regenerated and compared to the embedded ones.
\item If annotated smart contracts contain certifiable certificates, said
certificates are checked.
\end{enumerate}
\item This step could be computationally expensive, so it only needs to
be done the first time if the smart contract $SC$ is not modified.
\item Check if the verified annotations, proofs and/or certificates are
sufficient for the smart contract $SC$ to be declared correct, according
to the local policies of each party
\end{enumerate}
\item \textbf{Output:} \textit{true} if the verification was correct, \textit{false}
otherwise.\end{itemize}
\end{minipage}}\\

The previous protocols are combined in the following general execution
protocol for private and verifiable smart contracts:\\

\fbox{\begin{minipage}[t]{1\columnwidth}%
\textbf{Protocol 6.4: Realising Functionality 6.1 (Private and Verifiable
Smart Contracts)}
\begin{itemize}
\item \textbf{Parties:} $E_{1},E_{2},...,E_{N}$, set of nodes $N_{i}$
of a blockchain $B$
\item \textbf{Inputs:} smart contract $SC$, private inputs from parties
$E_{1}:\overrightarrow{x},E_{2}:\overrightarrow{y},...,E_{N}:\overrightarrow{z}$
\item \textbf{The protocol:}

\begin{enumerate}
\item The parties invoke Functionality 6.2 (Standard verification of smart
contracts), implemented using Protocol 6.3, and obtain the verification
status of the smart contract $SC$.
\item The parties invoke Functionality 4.1 (Secure computation of smart
contracts), implemented using Protocol 4.2 where party $E_{1}$ inputs
$\overrightarrow{x}$, party $E_{2}$ inputs $\overrightarrow{y}$
and the rest of parties $E_{N}$ their corresponding inputs. Parties
receive results $E_{1}:\overrightarrow{r_{1}}$, $E_{2}:\overrightarrow{r_{2}}$
and the rest of parties their corresponding outputs $E_{N}:\overrightarrow{r_{N}}$.
\end{enumerate}
\item \textbf{Outputs:} results from the secure computation $E_{1}:\overrightarrow{r_{1}},E_{2}:\overrightarrow{r_{2}},...,E_{N}:\overrightarrow{r_{N}}$.\end{itemize}
\end{minipage}}

\subsection{Security Analysis\label{sec:Security-Analysis}}

The security analysis follows the standard definition of static semi-honest
security in the standalone setting\cite{Goldreich04} and assumes
semi-honest security of basic building blocks: Yao's protocol\cite{cryptoeprint:2004:175}
and oblivious transfer accelerated with OT-extension\cite{EPRINT:ALSZ13,Ishai03extendingoblivious},
implying the security for the Protocol 4.2 realising Functionality
4.1. Additionally, ORAMs\cite{Goldreich:1987:TTS:28395.28416} could
also be used to speed up dynamic memory accesses: also note that when
the number of parties is huge, communication locality plays a central
role and may be better achieved with ORAMs\cite{EPRINT:BoyChuPas14,EPRINT:LuOst15}.
The following theorem proves the general protocol 6.3:
\begin{thm}
(General Protocol). Protocol 6.4 realises Functionality 6.1 against
static corruptions in the semi-honest security model augmented with
verifiability of the code.\end{thm}
\begin{proof}
Correctness and privacy of Protocol 4.2 (1.b) follows from Yao's protocol\cite{cryptoeprint:2004:175}
and oblivious transfer accelerated with OT-extension\cite{EPRINT:ALSZ13,Ishai03extendingoblivious};
in case of a multi-party setting, correctness and privacy follows
from one of the multi-party protocols (GMW\cite{gmw1987}, BMR\cite{bmr},
BGW\cite{BGW88,EPRINT:AshLin11}); when using ORAMs, the security
proof comes from the specific security proof of the chosen ORAM scheme
(Circuit ORAM\cite{EPRINT:WanChaShi14}; Square-Root ORAM\cite{revisitingSqORAM,originalORAM});
when using TLS, we make use of the composability of the security proofs
of said protocols\cite{EPRINT:GMPSS08,EPRINT:KMOTV14,EPRINT:BFKPSZ14}.
Regarding verifiability, it follows from Protocol 6.3 realising Functionality
6.2.
\end{proof}

\subsection{Extended verification of smart contracts}

Trusted third parties (e.g., governments, central banks, regulating
bodies) may provide specifications against which proofs must be generated.
Thus, Gyges attacks\cite{EPRINT:JueKosShi16} can be prevented if
parties and executing nodes require that all the executed smart contracts
must adhere to said specifications and be digitally signed by trusted
third parties: that is, prevent them from malicious hackers and others
engaged in unlawful business practices while being compliant with
certain regulations and laws (e.g., anti-money laundering laws). Comparing
the next two diagrams provides a better understanding of the introduced
differences:

\begin{figure}[H]
\centering{}\includegraphics[scale=0.65]{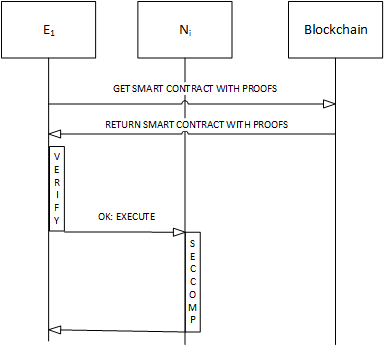}\protect\caption{Standard verification of smart contracts}
\end{figure}

\begin{figure}[H]
\centering{}\includegraphics[scale=0.65]{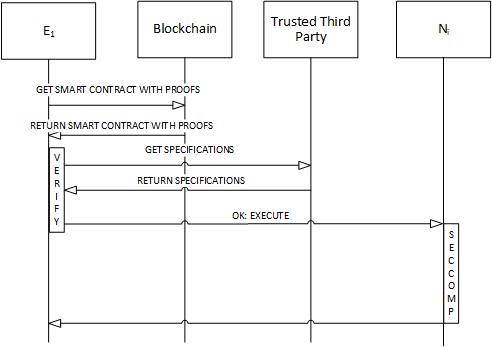}\protect\caption{Extended verification of smart contracts}
\end{figure}

To formalize this concept, the following Functionality 6.5 is introduced:

\fbox{\begin{minipage}[t]{1\columnwidth}%
\textbf{Functionality 6.5: Extended verification of smart contracts}
\begin{itemize}
\item \textbf{Parties:} $E_{1},E_{2},...,E_{N}$, trusted third parties
$T_{i}$
\item \textbf{Inputs:} smart contract $SC$
\item \textbf{The functionality:}

\begin{enumerate}
\item Parties $E_{1},E_{2},...,E_{N}$ obtain the annotations/proofs/certificates
of smart contract $SC$
\item Check digital signatures from trusted third parties $T_{i}$ and verify
conformance to their specifications: verify annotations, proofs and/or
certificates available on the smart contract $SC$
\item Check if the verified annotations, proofs and/or certificates are
sufficient for the smart contract $SC$ to be declared correct, according
to the local policies of each party
\end{enumerate}
\item \textbf{Output:} \textit{true} if the verification was correct, \textit{false}
otherwise.\end{itemize}
\end{minipage}}\\

The following Protocol 6.6 implements Functionality 6.5:\\

\fbox{\begin{minipage}[t]{1\columnwidth}%
\textbf{Protocol 6.6: Realising Functionality 6.5 (Extended verification
of smart contracts)}
\begin{itemize}
\item \textbf{Parties:} $E_{1},E_{2},...,E_{N}$, trusted third parties
$T_{i}$
\item \textbf{Inputs:} smart contract $SC$
\item \textbf{The protocol:}

\begin{enumerate}
\item Parties $E_{1},E_{2},...,E_{N}$ download the annotations/proofs/certificates
of smart contract $SC$ and check their digital signatures. Digital
signatures from trusted third parties $T_{i}$ are checked and conformance
to specifications from said third parties $T_{i}$ is tested: every
party $E_{i}$ has a security profile that specifies the mandatory/optional
trusted signing parties and specifications that every smart contract
must adhere to.

\begin{enumerate}
\item If annotated smart contracts do not contain any kind of proofs or
certificates, the annotations of the smart contract are automatically
tested using an heuristic/concolic execution engine. At least, it
should be digitally signed by a trusted third party.
\item If annotated smart contracts contain full/partial proofs, they are
locally regenerated and compared to the embedded ones. Some of said
proofs must conform to specifications from third parties.
\item If annotated smart contracts contain certifiable certificates, said
certificates are checked. Some of said certifiable certificates must
conform to specifications from third parties.
\end{enumerate}
\item This step could be computationally expensive, so it only needs to
be done the first time if the smart contract $SC$ is not modified.
\item Check if the verified annotations, proofs and/or certificates are
sufficient for the smart contract $SC$ to be declared correct, according
to the local policies of each party.
\end{enumerate}
\item \textbf{Output:} \textit{true} if the verification was correct, \textit{false}
otherwise.\end{itemize}
\end{minipage}}\\

Finally, the general Protocol 6.4 is reviewed with the new Protocol
6.6 for the extended verification of smart contracts:\\

\fbox{\begin{minipage}[t]{1\columnwidth}%
\textbf{Protocol 6.7: Realising Functionality 6.1 (Private and Extended
Verifiable Smart Contracts)}
\begin{itemize}
\item \textbf{Parties:} $E_{1},E_{2},...,E_{N}$, set of nodes $N_{i}$
of a blockchain $B$, trusted third parties $T_{i}$
\item \textbf{Inputs:} smart contract $SC$, private inputs from parties
$E_{1}:\overrightarrow{x},E_{2}:\overrightarrow{y},...,E_{N}:\overrightarrow{z}$
\item \textbf{The protocol:}

\begin{enumerate}
\item The parties invoke Functionality 6.5 (Extended verification of smart
contracts), implemented using Protocol 6.6, and obtain the extended
verification status of the smart contract $SC$.
\item The parties invoke Functionality 4.1 (Secure computation of smart
contracts), implemented using Protocol 4.2 where party $E_{1}$ inputs
$\overrightarrow{x}$, party $E_{2}$ inputs $\overrightarrow{y}$
and the rest of parties $E_{N}$ their corresponding inputs. Parties
receive results $E_{1}:\overrightarrow{r_{1}},E_{2}:\overrightarrow{r_{2}},...,E_{N}:\overrightarrow{r_{N}}$.
\end{enumerate}
\item \textbf{Outputs:} results from the secure computation $E_{1}:\overrightarrow{r_{1}},E_{2}:\overrightarrow{r_{2}},...,E_{N}:\overrightarrow{r_{N}}$.\end{itemize}
\end{minipage}}\\

The following theorem finally proves the extended general protocol
6.7:
\begin{thm}
(Extended General Protocol). Protocol 6.7 realises Functionality 6.1
against static corruptions in the semi-honest security model augmented
with extended verifiability of the code.\end{thm}
\begin{proof}
Correctness and privacy of Protocol 4.2 (1.b) follows from Yao's protocol\cite{cryptoeprint:2004:175}
and oblivious transfer accelerated with OT-extension\cite{EPRINT:ALSZ13,Ishai03extendingoblivious};
in case of a multi-party setting, correctness and privacy follows
from one of the multi-party protocols (GMW\cite{gmw1987}, BMR\cite{bmr},
BGW\cite{BGW88,EPRINT:AshLin11}); when using ORAMs, the security
proof comes from the specific security proof of the chosen ORAM scheme
(Circuit ORAM\cite{EPRINT:WanChaShi14}; Square-Root ORAM\cite{revisitingSqORAM,originalORAM});
when using TLS, we make use of the composability of the security proofs
of said protocols\cite{EPRINT:GMPSS08,EPRINT:KMOTV14,EPRINT:BFKPSZ14}.
Regarding verifiability against Gyges attacks\cite{EPRINT:JueKosShi16},
it follows from Protocol 6.6 realising Functionality 6.5.
\end{proof}

\subsection{Malicious security}

Smart contracts are compiled to efficient Boolean circuits that could
be executed on multiple secure computation frameworks with the added
benefit of obtaining malicious security\cite{cryptoeprint:2016:762,cryptoeprint:2017:030,cryptoeprint:2017:189,emp-toolkit,cryptoeprint:2017:214}:
protocols and theorems above can trivially be updated with their security
definitions and proofs. In case of using ORAMs, security against malicious
adversaries comes from specialized versions of said protocols \cite{EPRINT:AHMR14,EPRINT:Miao16,EPRINT:HazYan16}.

\subsection{Private Function Evaluation}

In this setting, the smart contract itself is kept secret: that is,
only one of the parties knows the function $f\left(x\right)$ computed
by the smart contract, whereas other parties provide the input to
the private function without learning about $f$ besides the size
of the circuit defining the function and the number of inputs and
outputs. Private function evaluation can be implemented using secure
function evaluation\cite{af90,practicalUC} by securely evaluating
a Universal Circuit\cite{Valiant:1976:UC:800113.803649} that is programmed
by the party knowing the function $f\left(x\right)$ to evaluate it
on the other parties' inputs: the security follows from that of the
secure function evaluation protocol that is used to evaluate the Universal
Circuit.

Recently, Valiant's Universal Circuit has been implemented and proven
practical\cite{EPRINT:KisSch16,cryptoeprint:2017:798}: its optimal
size was proven to be $\Omega\left(k\log k\right)$\cite{Valiant:1976:UC:800113.803649}
and its recent implementation further improved it by (at least) $2k$.
Although it's a considerable blowup in size, late MPC protocols provide
very efficient pre-processing phases that could ameliorate execution
times: online times are only 1-4\% of the total execution time in
the WAN setting\cite{cryptoeprint:2017:189,emp-toolkit,cryptoeprint:2016:1066},
and 2-5\% in the LAN setting\cite{cryptoeprint:2017:214,improvementsBMR}.
Since the generated universal circuit description $UC_{u,v,k^{*}}$
is public to all parties\cite{EPRINT:KisSch16,cryptoeprint:2017:798},
both function-dependent and function-independent pre-processing\cite{cryptoeprint:2017:189,emp-toolkit,cryptoeprint:2016:1066,cryptoeprint:2017:214,cryptoeprint:2017:862,improvementsBMR}
could be used in the setting of private function evaluation for maximum
speedup.

\section{Economic Impact and Legal Analysis}

Previous legal analysis\cite{practiceLaw} has considered how secure
computation fits current legal frameworks, although it's much more
interesting to evaluate how it could alter legal frameworks and its
economic impact.

\subsection{Data Privacy as Quasi-Property}

Quasi-property interests\cite{quasiProperty} refer to situations
in which the law seeks to simulate the functioning of property’s exclusionary
apparatus, through a relational entitlement mechanism, by focusing
on the nature and circumstances of the interaction in question, which
is thought to merit a highly circumscribed form of exclusion. The
term first appeared in a SCOTUS decision in International News Service
v. Associated Press\cite{INSvsAP}, in which Justice Pitney recognized
the right of an information gatherer to prevent a competitor from
free riding on the original gatherer’s labor for a limited period
of time: more specifically, it's limited in that it would only ever
exist between the two parties in question and never in the abstract
against the world at large.

Quasi-property allows to effectively treat data privacy as a property
right\cite{privacyAsQuasiProperty} and it's the closest analogue
in American law that grants individuals a property right in their
personal data, as privacy has always been considered in the privacy-preserving
literature in cryptography: in fact, the quasi-property view of privacy-preserving
computation has stronger grounding than European database rights\cite{rightsInData}.
This approach provides a common framework that underlies all of the
privacy torts, while avoiding the need to define privacy in such a
way that it describes every injury that the law recognizes as an invasion
of privacy and that can be generally categorized in four harms: (1)
information collection; (2) information processing; (3) information
dissemination, and (4) invasion.

A general discussion of the interpretation of privacy as property
from other points of view could be found in \cite{privacyAsIntellectualProperty}.

\subsection{A Solution to Arrow's Paradox}

Arrow's Paradox\cite{RePEc:nbr:nberch:2144} states that ``there
is a fundamental paradox in the determination of demand for information;
its value for the purchaser is not known until he knows the information,
but then he has in effect acquired it without cost''. \textit{Ex-ante},
the purchaser cannot value the original information since it can only
be known after it has been revealed; \textit{ex-post}, the purchaser
could not compensate the seller and disseminate it for free. Due to
the inherent properties of information (non-excludable, non-rivalrous),
markets for information cannot exist in the absence of intellectual
property rights\cite{RePEc:eee:respol:v:32:y:2003:i:2:p:333-350}
because the original producer/inventor of any information loses the
monopoly on it after the information is revealed; regarding financial
information, the finance literature generally agrees in that the troubles
in informational trading explain financial intermediation\cite{RePEc:bla:jfinan:v:32:y:1977:i:2:p:371-87,RePEc:oup:restud:v:51:y:1984:i:3:p:415-432.,RePEc:eee:jfinin:v:1:y:1990:i:1:p:3-30,RePEc:ecm:emetrp:v:58:y:1990:i:4:p:901-28}.

Secure computation techniques offer a practical solution to Arrow's
Paradox: distrustful third-parties can compute on private information
without disclosing it, allowing its valuation while preventing intellectual
property theft; this is significantly better than previously known
methods based on partial the revelation of information\cite{RePEc:aea:aecrev:v:84:y:1994:i:1:p:190-209,RePEc:oup:restud:v:69:y:2002:i:3:p:513-531,RePEc:tpr:jeurec:v:3:y:2005:i:2-3:p:745-754}.
More generally, secure computation techniques make information excludable
while maintaining its non-rivalrousness: that is, the original producer/inventor
of the information retains market power over it while maintaining
its costless resale (i.e., the marginal cost of an additional digital
copy is zero), theoretically allowing for digital goods of infinite
value that bypass Coase conjecture\cite{RePEc:ucp:jlawec:v:15:y:1972:i:1:p:143-49}
since it's also possible to prevent that first time purchasers resell
the information they have queried/acquired\cite{RePEc:eee:gamebe:v:2:y:1990:i:4:p:337-361,RePEc:eee:gamebe:v:3:y:1991:i:3:p:339-349},
although in practice having to account for the costs of the slowdown
introduced by secure computation techniques.

\subsection{Expansion of Trade Secrecy}

While software copyright and patentability protection have been weakened,
trade secrecy stands firm. As defined in the United States' Uniform
Trade Secrets Act\cite{utsa}: ``Trade secret means information,
including a formula, pattern, compilation, program, device, method,
technique, or process, that: (i) derives independent economic value,
actual or potential, from not being generally known to, and not being
readily ascertainable by proper means by, other persons who can obtain
economic value from its disclosure or use, and (ii) is the subject
of efforts that are reasonable under the circumstances to maintain
its secrecy''. The definition on the European Union's Trade Secrets
Directive\cite{utsa} implicitly includes computer programs, financial
innovations, lists of customers, business statistics and many other
types of secret information.

There is no definitional uncertainty in trade secret law: unlike copyright,
there is no exclusion for functionality; unlike patents, there is
no exclusion for abstract ideas\cite{hiddenPlainSight}. There isn't
any uncertainty about the exact contours of their extension either,
such as in the debate of the mutual exclusivity of copyright and patents
protections\cite{copyrightPatentOverlap}. In some cases, widely distributed
software may remain a trade secret if the license agreement requires
confidentiality and return upon non-use (see Data Gen. Corp. v. Grumman
Systems Support Corp.\cite{DGCvsGSSC}).

Although software might be reverse engineered (i.e., an acceptable
way to discover a trade secret that prevents its general use for protecting
software), secure computation techniques effectively hinder and/or
completely forbid reverse engineering. Trade secrets can be justified
as a form, not of traditional property, but of intellectual property
in which secrecy is central\cite{virtuesTradeSecrecy} and may serve
the purposes of IP law better than more traditional IP rights. Stronger
trade secrecy law increases R\&D in high technology industries\cite{innovationTradeSecrecy},
reduces patenting\cite{patentsTradeSecrecy} and increases trade flows\cite{tradeSecrecyEconomics}:
similar effects are expected from the utilisation of the technologies
described in the present publication.

\subsection{Verifiability and Self-Enforcement}

The term \textit{Lex Cryptographia}\cite{lexCryptographia} designates
a new body/subset of law, containing rules administered through smart
contracts and decentralized autonomous organizations: this concept
is inspired on \textit{Lex Mercatoria}, an old subset of customs that
became recognized as a customary body of law for international commerce.
Smart contracts are just software programs with very specialized functionalities
intended to replace paper contracts: the practice of law could thus
follow the path of software, with smart contract programming languages
becoming more powerful and easier to develop, transforming the legal
profession with more technical lawyers. By design, smart contract
cannot be breached: once contracting parties have agreed to be bound
by a particular clause, the code's immutability binds them to that
clause without leaving them the possibility of a breach; that is,
the code defines its own interpretation and enforces the defined rules
contained on it without the need of third parties (i.e., self-enforcement).
The only way to escape from contractual obligations that the parties
no longer want to honor is by including legal provisions into smart
contract's code. Over time, law and code may converge, so that infringing
the law will be effectively breaking the code: that is, \textit{Lex
Cryptographia} is stronger than \textit{Lex Posita} and will become
more prevalent.

The strictness of self-enforcement has been much criticized\cite{levy2017book,ohara2017smart}:
it could be misused and turned against the contracting parties and
it doesn't represent the realities of real-world enforcement of contracts.
Adding verifiability to smart contracts as proposed in this work prevents
all these problems as it allows to check conformity to the specifications
from third parties (e.g., governments, regulating bodies, standards).
In sum, contracting parties are private lawmakers\cite{computableContracts}
in control of both the substance and the form of their contractual
obligations (i.e., \textit{pacta sunt servanda}), but third parties
could intervene to regulate said private agreements by redacting specifications
that smart contracts must formally verify against.

\subsection{Markets for Smart Contracts\label{sub:Markets-for-Smart-Contracts}}

Paradoxically, smart contracts are hardly bought/sold: the inexistence
of strong property rights hampers the development of markets on smart
contracts; actually, not even companies developing and operating them
are being acquired/merged based on the value of their smart contracts
since they can easily be reverse engineered and cloned.

Private smart contracts provide strong property rights based on cryptographic
techniques which allow for the emergence of markets to trade them.
Depending on the cryptographic techniques employed, the following
advantageous situations could be considered:
\begin{enumerate}
\item The algorithms contained within smart contracts could be traded without
disclosing their details when using private function evaluation, homomorphic
encryption and/or indistinguishability obfuscation.
\item More practically, when using garbled circuits or secret sharing techniques
the source of value will reside on encrypted data processed by the
algorithms of private smart contracts, these cryptographic techniques
being much more efficient than homomorphic encryption/IO; said encrypted
data could be stored on the blockchain, easing the transfer of property
of the smart contract.
\end{enumerate}
Traditionally, binary compilation and obfuscation have been used to
safeguard the value of software: secure computation techniques offer
a provably-secure way to protect the value of even purely open-source
software. Furthermore, the ability to safely trade private smart contracts
will justify their higher development costs.

\subsection{Effects on Currency Competition}

Latest analysis on the monetary policy of cryptocurrencies\cite{currencyCompetition}
provide insights on the effects of currency competition:
\begin{itemize}
\item Monetary equilibrium between private cryptocurrencies will not deliver
price stability: profit-maximizing entrepreneurs issuing cryptocurrencies
do not have real incentives to provide stable currencies, only to
maximise their seigniorage.
\item Monetary systems consisting of only private cryptocurrencies in the
equilibrium with stable prices do not provide the socially optimum
quantity of money: competition between cryptocurrencies is not enough
to provide optimal outcomes since entrepreneurs do not internalise
the pecuniary externalities by minting additional tokens.
\item Unlike private money, government money has fiscal backing because
it can tax agents in the economy. But in competition with people willing
to hold cryptocurrencies, the implementation of monetary policy in
deflationary settings will be significantly impaired since profit-maximizing
entrepreneurs will be unwilling to retire their private currencies
and instead choose to increase their issued money.
\item Government money could co-exist without intervention in a unique equilibrium
with cryptocurrencies if the minted cryptocurrency growths following
a predetermined algorithm (e.g., Bitcoin) and the proceedings are
used to buy/finance sufficiently productive capital. 
\end{itemize}
Additionally, social efficiency may also be achieved with different
cryptocurrencies featuring diverse functionalities that provide market
power to their issuers and users: particularly, this is the case of
the private and verifiable smart contracts of this paper, since they
could be used to provide natural monopolies on the encrypted programs
stored on them.

\subsection{Token vs. Account-based Cryptocurrencies}

Most cryptocurrencies are based on the model of issuing and transacting
tokens, using some form of distributed ledger to keep track of the
ownership of said tokens (e.g., Bitcoin, Ethereum and Zcash): they're
the digital equivalents of cash.

Another unexplored model of cryptocurrency is that of holding funds
in accounts at the central bank or in depository institutions, resembling
debit cards: in fact, users of Digital Currency Exchanges maintain
accounts holding substantial amounts of wealth, but in the form of
token-based cryptocurrencies. There are many efficiency gains to be
expected from account-based cryptocurrencies: latest macroeconomic
models\cite{RePEc:boe:boeewp:0605} show that they would permanently
raise GDP by as much a 3\% in the USA due to lower bank funding costs,
lower monetary transaction costs and lower distortionary taxes; additionally,
they would introduce new tools for the central bank to stabilise the
business cycle.

The technology required to implement account-based cryptocurrencies
is different and more complex than the required for token-based cryptocurrencies:
the smart contracts proposed in this paper are a perfect fit for this
task, due to their ability to maintain privacy and guarantee their
perfect functionality through formal verification techniques.

\subsection{Impact on Market Structures}

Although disintermediation will be one of the first consequences of
the application of cryptographic smart contracts to financial markets,
as happened in the past with the introduction of the Internet\cite{RePEc:wop:pennin:00-35},
bank and fund concentrations may also increase: information asymmetry
is associated with more concentration\cite{JOFI:JOFI1219} and information
production and its hiding is a valuable activity of banks, that is,
opacity has value in itself\cite{RePEc:aea:aecrev:v:107:y:2017:i:4:p:1005-29}.
\textit{Ceteris paribus}, it's difficult to estimate the resulting
equilibria of the impact of this technology on the level of concentration
in the financial industry, given the high number of inter-related
variables; nonetheless it's easier to predict that practices like
shadow banking\cite{shadowBanking} will increase, as new financial
technology accounts for 35\% of their growth\cite{shadowBankIncrease}.

\subsection{The Valuation of Secrecy and the Privacy Multiplier}

The use of private smart contracts (see Section \ref{sec:Private-Smart-Contracts})
has inmediate economic benefits for the secured information, because
when information is withholden to prevent data leakages then its valuation
is discounted. This secrecy discount is given by

\[
S_{T}=e^{-yT}\left(2\varPhi\left(\frac{\sigma\sqrt{T}}{2}\right)-1\right)
\]
with $T$ denoting the time to maturity, $\varPhi$ the cumulative
normal distribution of mean 0 and standard deviation 1, $y$ is the
constant continuously compounded dividend yield and $\sigma$ the
constant volatility of the price of the secret data. However, private
smart contracts\ref{sec:Private-Smart-Contracts} not only allow the
private appropriation of the secrecy discount, they also enable that
the same data can be sold multiple times without the holder exhausting
its exclusivity, that is, they introduce a multiplier effect on its
valuation (i.e., Privacy Multiplier). For more details on these concepts
and derivation of formulae, refer to \cite{secrecyAndPrivacyMultiplier}.

\section{Related work}

A number of related works using cryptographic techniques are as follows:
\begin{itemize}
\item Enigma\cite{1506.03471,enigmaThesis} is coded in Python, a language
lacking verification libraries/toolkits and formal semantics so it
can't be used for proof-carrying code; its cryptographic protocols
are not constant-round, thus settings with some latency will be it
excruciatingly slow (i.e., Internet); finally, the only supported
distributed ledger is Bitcoin. Later developments show that Enigma
is pivoting to a decentralized data marketplace using deterministic
and order-preserving encryption\cite{EnigmaDataMarketplace}.
\item WYS{*}\cite{fstarmpc} offers an elegant solution for verifiable and
and secure multi-party computation based on the dependently typed
feature of the F{*} language. Unfortunately, the use of said research
language also compromises its real-world adoption; additionally, it's
not integrated on any blockchain and does not offer proof-carrying
code. Another similar work\cite{EPRINT:ABBDDG14} uses EasyCrypt to
verify Yao's garbled circuits.
\item Hawk\cite{EPRINT:KMSWP15} focuses on protecting the privacy of transactions
using zk-SNARKs: later works like {[}Solidus\cite{cryptoeprint:2017:317},
Confidential Transactions\cite{confidentialTransactions}, Bolt\cite{EPRINT:GreMie16}{]}
provide secure transactions at the protocol level with more efficient
techniques. It only mentions secure multi-party computation as a tentative
way to replace the trusted auction manager, rejecting it as impractical.
\item Chainspace\cite{chainspace} is a sharded smart contract platform
that includes some examples of private smart contracts based on the
PETlib library using homomorphic encryption, and zero-knowledge proofs
only after the execution of the smart contract.
\item ZeroCash\cite{EPRINT:BCGGMT14} and Monero\cite{EPRINT:Noether15}
provide security for transactions, not smart contracts. Note that
secure transactions on blockchains are just a special restricted case
of smart contracts enabled with cryptography, but not the other way
around.
\item Another concurrent work\cite{cryptoeprint:2018:404} provided formal
verification over a MPC framework (i.e., Dafny over the SecreC language
of Sharemind), but only to formally prevent and justify the leakage
of secrets in secure multi-party computation.
\item Oyente\cite{EPRINT:LCOSH16}, a symbolic execution tool to formally
verify Ethereum smart contracts (EVM).
\item Accountable algorithms\cite{accountableAlgorithms} propose a commit-and-prove
protocol with zk-SNARKs to provide accountability proofs of compliance
to legal standards without revealing key attributes of computerized
decisions after said decisions have been taken, but not before their
execution as is done in the present paper.
\item Previous works on combining MPC with Bitcoin\cite{cryptoeprint:2013:784,EPRINT:ADMM13b}
use it as support to obtain fairness in MPC, and not better smart
contracts.
\item Previous projects\cite{EPRINT:MEKHL12,EPRINT:ABBKSS10} designed high-level
languages for Zero-Knowledge Proofs of Knowledge but not for Zero-Knowledge
Proofs of Proofs, and their languages were restricted and not general
purpose. 
\end{itemize}

\section{Conclusions, subsequent and future work}

The present paper has tackled and successfully solved the problem
of improving the privacy, correctness and verifiability of smart contracts,
resolving the DAO and Gyges attacks. Examples have been shown to demonstrate
its practical viability.

\subsection{Subsequent work}

The ability to save the state of garbled circuits and restore them
at later times is a major improvement, bringing them closer to secret
sharing techniques:
\begin{itemize}
\item partial garbled circuits\cite{reuseloseit}: for each wire value,
the generator sends two values to the evaluator, transforming the
wire labels of the evaluator to another garbled circuit; depending
on its point and permute bit, the value from a previous garbled circuit
computation is mapped to a valid wire label in the next computation.
\item reactive garbled circuits\cite{EPRINT:NieRan15}: a generalization
of garbled circuits which allows for partial evaluation and dynamic
input selection based on partial outputs.
\item reusable garbled circuits\cite{EPRINT:GKPVZ12,cryptoeprint:2016:654}
allow for token-based obfuscation where the code producer provides
tokens to code consumers representing rights to execute garbled smart
contracts: later constructions\cite{Wang2017} improve their concrete
efficiency.
\end{itemize}
Raziel is ongoing development and subject to improvements.

\subsection{Future work}

Some encrypted smart contracts will be perpetual (e.g., consols):
if they ever store any kind of encrypted secure computation (e.g.,
secret shares, garbled circuits, homomorphic encryptions, IO) eventually
there will be the need to update their encrypted contents to upgrade
their security level, a concept that has already been considered\cite{EPRINT:LCLPY13,EPRINT:AnaCohJai16,cryptoeprint:2017:137,cryptoeprint:2018:118}.

MPC and SGX are not mutually exclusive and they will be used jointly
to obtain better performance\cite{cryptoeprint:2016:1057,Gupta2016,cryptoeprint:2016:1027}.
On the other hand, it's difficult to derive realistic threat models
and abstractions\cite{cryptoeprint:2016:014,cryptoeprint:2016:1027,cryptoeprint:2017:565}
that withstand the latest attacks against SGX \cite{1702.07521,1702.08719,1703.06986,asyncshock,Xu:2015:CAD:2867539.2867677,tsgx,1611.06952,rollbackSGX,cryptoeprint:2017:736,203696,intelSGXBug,SGXnoPageFaults,jang:sgx-bomb,DBLP:journals/corr/XiaoLCZ17,leakyCauldron,SGXStep,spectreSGX,1802.09085,cachequote,vanbulck2018foreshadow,guarddilemma,nemesis,2018arXiv181105441C,splitSpectre,SMoTherSpectre,ZombieLoad2019,ridl}.

\section*{Acknowledgments}

I would like to thank David Evans and Jonathan Katz for helpful comments
on the paper.

{\footnotesize{}\bibliographystyle{alpha}
\bibliography{crypto,abbrev2,bib}
}{\footnotesize \par}

\appendix

\section{Outsourcing Secure Computations for Cloud-based Blockchains\label{sec:Appendix-outsource}}
\begin{thm}
(Outsourcing for Cloud-Based Blockchains). Assuming secure channels
between the parties $E_{1},E_{2},...,E_{N}$ and the nodes $N_{i}$
of the blockchain $B$, Protocol B.2 securely realises Functionality
B.1 against static corruptions in the semi-honest security model with
the garbling scheme satisfying privacy, obliviousness and correctness.\end{thm}
\begin{proof}
Two cases must be analyzed: corruption at the node $N_{E}$ executing
the secure computation and corruption of the parties.\\

\textbf{Corruption at node $N_{E}$. }The simulator $Simul_{N_{E}}\left(1^{\lambda},result_{N_{E}}\right)$,
where $result_{N_{E}}$ stores the $KeyExchange$, $SendPrivateParameters$
and $SECCOMP$ run by each party $E_{1},E_{2},...,E_{N}$:
\begin{itemize}
\item $KeyExchange\left\langle E_{i}\left(1^{\lambda}\right),N_{E}\right\rangle $:
sample a public key $pk_{i}$ for party $E_{i}$.
\item $SendPrivateParameters\left\langle E_{i},N_{E}\right\rangle $: sample
a vector of random values $c_{i}=\left(\left(X_{i1}^{x_{i}\left[1\right]},...,X_{il}^{x_{i}\left[l\right]}\right),n_{i}\right)$.
\item $SECCOMP\left\langle E_{k},N_{E}\right\rangle \left(SC\right)$: a
simulated garbled circuit and garbled values are computed. Then, the
simulator computes one entry of the pivot tables by encrypting random
values and the other entry by encrypting each garbled value with the
random values generated in the $SendPrivateParameters$ phase:

\begin{itemize}
\item A simulator of garbled circuits generates a simulated garbled circuit
$GC_{SC}$ and garbled values $w_{jl}$ for each party $E_{j}$ and
index $l$ of the length of $\overrightarrow{x_{i}}$.
\item Compute $Enc_{s_{jl}}\left(0\right)$ with random key $s_{jl}$ and
$Enc_{X_{jl}^{x_{j}\left[l\right]}}\left(w_{jl}\right)$ and save
them into pivot table $P_{q}\left[j,l\right]$ in random order.
\item Return the pivot tables $P_{q}$ for each $E_{j}$, the garbled circuit
$GC_{SC}$ and values $w_{jl}$.
\end{itemize}
\end{itemize}
We prove using hybrid arguments that the view generated by $Simul_{N_{E}}$
is indistinguishable from the view obtained by $N_{E}$ in the real
world:\\

$H_{0}.$ The real world view is computed using the real inputs of
the parties and according to the original protocol.\\

$H_{1}^{i}.$ What is different between $H_{1}^{i}$ and $H_{1}^{i+1}$
is that for each $E_{i}$ in $H_{1}^{i}$ the inputs are encoded using
random values: that is, $X_{il}^{x_{i}\left[l\right]}=rand_{il}$.
It would be possible to obtain a distinguisher for the pseudorandomness
of $PRF$ assuming a distinguisher between $H_{1}^{i}$ and $H_{1}^{i+1}$.
The values $X_{il}^{x_{i}\left[l\right]}$ will be computed with an
oracle by the reduction: the view will be distributed as in game $H_{1}^{i}$
if the oracle is using a random function; otherwise, it will be distributed
as in game $H_{1}^{i+1}$ if the oracle is a pseudo-random function.
Thus, a contradiction will be reached since any adversary distinguishing
$H_{1}^{i}$ from $H_{1}^{i+1}$ with non-negligible probability will
be able to break the pseudo-randomness with the same probability.
\\
The sequence of hybrid games starts with $H_{1}^{1}$ where random
values encode inputs of party $E_{i}$, to $H_{1}^{n}$ where random
values from parties $E_{1}...E_{N}$ encode inputs of the parties.\\
Finally, note that the hybrid $H_{0}=H_{1}^{0}$ corresponds to the
case where all inputs are pseudo-random values, while the last game
$H_{1}^{n}$ corresponds to the case in which random values encode
all inputs.\\

$H_{2}^{i,j,l}.$ In this hybrid, the call $SECCOMP$ by party $E_{i}$
computes the pivot table $P_{q}\left[j,l\right]=Enc_{s_{jl}^{x_{j}[l]}}\left(0\right),Enc_{s_{jl}^{x_{j}[l]}}\left(w_{jl}^{x_{j}[l]}\right)$,
that is, only one garbled value per wire is encrypted: therefore,
a contradiction will be reached since an adversary distinguishing
between $H_{2}^{i,j,l}$ and $H_{2}^{i,j,l+1}$ with non-negligible
probability can be reduced to a distinguisher for the indistinguishability
under chosen-plaintext attacks (IND-CPA) of the encryption scheme.\\
The sequence of hybrid games starts with $H_{2}^{1,1,l}$ where the
pivot table of party $E_{i}$ is encoded with only one garbled value
per wire, to $H_{2}^{n,n,l}$ where all pivot tables of parties $E_{1}...E_{N}$
are encoded with only on garbled value per wire.\\

$H_{3}^{j}.$ The simulator $Simul_{N_{E}}$ computes the pivot tables
and a simulator of garbled circuits is used instead of the real garbling
scheme. It would be possible to obtain a distinguisher for garbled
circuits assuming a distinguisher between $H_{2}^{n,n,l}$ and $H_{3}^{j}$:
the inputs to the garbled circuits circuit will be computed by an
oracle by the reduction as in $H_{2}^{n,n,l}$, to get garbled values
and garbled circuits.\\
Finally, we reach the ideal world with hybrid $H_{3}^{n}$, finishing
the proof that started with the real world game $H_{0}$.\\

\textbf{Corruption of the parties}. Assuming a single party $E_{i}$
is corrupted and secure channels between the parties and $N_{E}$,
the simulator $Simul_{E_{i}}$ is defined as:
\begin{itemize}
\item $KeyExchange\left\langle E_{i}\left(1^{\lambda}\right),N_{E}\right\rangle $:
a public key $pk_{i}$ is chosen by party $U_{i}$.
\item $SendPrivateParameters\left\langle E_{i},N_{E}\right\rangle $: using
the honest version for $\overrightarrow{x_{i}}$, honestly compute
values $X_{il}^{\overrightarrow{x_{i}}\left[l\right]}=PRF_{k_{S}}\left(\overrightarrow{x_{i}},l,n_{i}\right)$.
\item $SECCOMP\left\langle E_{k},N_{E}\right\rangle \left(SC\right)$: honestly
compute garbled circuit and pivot tables. For the result, select the
garbled values according to Functionality B.1 called with value $x_{i}$.
\end{itemize}
Note that the method to compute the result is the only difference
in computing the view of $E_{i}$: while the evaluation of the garbled
circuit is required in the real world, in the ideal world the correct
resulting garbled values are chosen by the simulator $Simul_{E_{i}}$
using the result of Functionality B.1. Therefore, the indistinguishability
of the view of $E_{i}$ is due to the secure channels: the case for
more corrupted parties trivially follows from this argument.\end{proof}

\end{document}